\documentclass[11pt, english]{article}

\usepackage[margin=1in]{geometry}

\usepackage{varioref}
\usepackage[usenames,svgnames,xcdraw,table]{xcolor}
\definecolor{DarkBlue}{rgb}{0.1,0.1,0.5}
\definecolor{DarkGreen}{rgb}{0.1,0.5,0.1}

\usepackage{nicefrac}

\usepackage{hyperref}
\hypersetup{
   colorlinks   = true,
 linkcolor    = DarkBlue, 
 urlcolor     = DarkBlue, 
	 citecolor    = DarkGreen 
}

\usepackage{appendix}

\usepackage{mathtools}

\usepackage{algorithmic}
\usepackage{algorithm}
\usepackage{array}

\usepackage[font={small,it}]{caption}

\usepackage{graphicx,epsfig,amsmath,latexsym,amssymb,verbatim}

\usepackage[square,sort,semicolon,numbers]{natbib}

\newcommand{\extra}[1]{}

\usepackage{amsmath,amssymb,amsfonts}
\usepackage{amsthm}
\usepackage{thmtools,thm-restate}
\usepackage{enumitem}

\newtheorem{theorem}{Theorem}

\newtheorem{definition}[theorem]{Definition}
\newtheorem{lemma}[theorem]{Lemma}

\newtheorem{proposition}[theorem]{Proposition}
\newtheorem{claim}{Claim}

\def\squareforqed{\hbox{\rlap{$\sqcap$}$\sqcup$}}
\def\qed{\ifmmode\squareforqed\else{\unskip\nobreak\hfil
\penalty50\hskip1em\null\nobreak\hfil\squareforqed
\parfillskip=0pt\finalhyphendemerits=0\endgraf}\fi}
\def\endenv{\ifmmode\;\else{\unskip\nobreak\hfil
\penalty50\hskip1em\null\nobreak\hfil\;
\parfillskip=0pt\finalhyphendemerits=0\endgraf}\fi}

\renewenvironment{proof}{\noindent \textbf{{Proof~} }}{\qed\medskip}
\newenvironment{proof+}[1]{\noindent \textbf{{Proof #1~} }}{\qed\medskip}

\mathchardef\ordinarycolon\mathcode`\:
\mathcode`\:=\string"8000
\def\vcentcolon{\mathrel{\mathop\ordinarycolon}}
\begingroup \catcode`\:=\active
  \lowercase{\endgroup
  \let :\vcentcolon
  }

\newcommand{\E}{\mathbb{E}}

\newcommand{\NSW}{\mathrm{NSW}}
\newcommand{\XOS}{\textsc{XOS}}

\DeclareMathOperator*{\argmax}{arg\,max}

\usepackage{amsmath}
\usepackage{amssymb,amsthm,mathtools}

\usepackage{amsmath,amsfonts} 
\usepackage{caption} 
\usepackage{xcolor}
\usepackage{hyperref}
\hypersetup{
   colorlinks   = true,
 linkcolor    = DarkBlue, 
 urlcolor     = DarkBlue, 
	 citecolor    = DarkGreen 
}

\usepackage[nottoc]{tocbibind}
\usepackage{comment}
\usepackage{xfrac}

\usepackage{verbatim}

\usepackage[font=small,labelfont=bf]{caption}
\usepackage{xspace}
\usepackage{etoolbox}

\newcommand*{\textcal}[1]{%
  \textit{\fontfamily{qzc}\selectfont#1}%
}

\newcommand{\A}{\mathcal{A}}
\newcommand{\G}{\mathcal{G}}
\newcommand{\I}{\mathcal{I}}
\newcommand{\N}{\mathcal{N}}
\newcommand{\Q}{\mathcal{Q}}

\newcommand{\R}{{\rm R}}

\newcommand{\oveN}{\overline{N}}

\newcommand{\instance}[1]{\langle {#1} \rangle}
\newcommand{\val}{\{v_i\}_{i \in [n]}}

\newcommand{\M}{{\rm M}}

\newcommand{\constMK}{16}
\newcommand{\constTone}{256}

\newcommand{\constTthree}{16}

\newcommand{\constlinear}{64}
\allowdisplaybreaks

\title{\bfseries Sublinear Approximation Algorithm for Nash Social Welfare with XOS Valuations}

\author{Siddharth Barman\thanks{Indian Institute of Science. {\tt barman@iisc.ac.in}} \quad Anand Krishna\thanks{Indian Institute of Science. {\tt anandkrishna@iisc.ac.in}} \quad Pooja Kulkarni\thanks{University of Illinois at Urbana-Champaign. {\tt poojark2@illinois.edu}} \quad Shivika Narang\thanks{Indian Institute of Science. {\tt shivika@iisc.ac.in}}}
\date{}

\begin{document}
\maketitle
\begin{abstract}
We study the problem of allocating indivisible goods among $n$ agents with the objective of maximizing Nash social welfare ($\NSW$). This welfare function is defined as the geometric mean of the agents' valuations and, hence, it strikes a balance between the extremes of social welfare (arithmetic mean) and egalitarian welfare (max-min value). Nash social welfare has been extensively studied in recent years for various valuation classes. In particular, a notable negative result is known when the agents' valuations are complement-free and are specified via value queries: for $\XOS$ valuations, one necessarily requires exponentially many value queries to find any sublinear (in $n$) approximation for $\NSW$. Indeed, this lower bound implies that stronger query models are needed for finding better approximations. Towards this, we utilize demand oracles and $\XOS$ oracles; both of these query models are standard and have been used in prior work on social welfare maximization with $\XOS$ valuations. 

We develop the first sublinear approximation algorithm for maximizing Nash social welfare under $\XOS$ valuations, specified via demand and $\XOS$ oracles. Hence, this work breaks the $O(n)$-approximation barrier for $\NSW$ maximization under $\XOS$ valuations. We obtain this result by developing a novel connection between $\NSW$ and social welfare under a capped version of the agents' valuations. In addition to this insight, which might be of independent interest, this work relies on an intricate combination of multiple technical ideas, including the use of repeated matchings and the discrete moving knife method. In addition, we partially complement the algorithmic result by showing that, under $\XOS$ valuations, an exponential number of demand and $\XOS$ queries are necessarily required to approximate $\NSW$ within a factor of $\left(1 - \frac{1}{e}\right)$. 
\end{abstract}
\thispagestyle{empty}

\clearpage

\tableofcontents
\thispagestyle{empty}

\clearpage

\setcounter{page}{1}

\section{Introduction}
The theory of fair division has been extensively studied over the past several decades in mathematical economics \cite{brams1996fair,Moulin03} and, more recently, in computer science \cite{brandt2016handbook}. At the core of this vast body of work lies the question of finding fair and economically efficient allocations. Through the years and for various settings, different notions of fairness and economic efficiency have been defined \cite{Moulin03}. In particular, social welfare (defined as the sum of the valuations of the agents) is a standard measure of economic efficiency. On the other hand, egalitarian welfare (defined as the minimum value across the agents) is a well-established fairness criterion.\footnote{Note that the average social welfare corresponds to the arithmetic mean of values and egalitarian welfare is the minimum value.} Indeed, these two welfare objectives are not necessarily compatible; an allocation with high social welfare can have very low egalitarian welfare, and vice versa. 

A meaningful compromise between the extremes of economic efficiency and fairness is achieved through the Nash social welfare ($\NSW$). This welfare function is defined as the geometric mean of the agents' valuations and it strikes a balance between between the arithmetic mean (average social welfare) and the minimum value (egalitarian welfare).  


The Nash social welfare is known to satisfy fundamental axioms, including the Pigou-Dalton transfer principle, Pareto dominance, symmetry, and independence of unconcerned agents \cite{Moulin03}. In fact, up to standard transformations, $\NSW$ is characteristically the unique welfare function that satisfies scale invariance along with particular fairness axioms \cite{Moulin03}. The efficiency and fairness properties of $\NSW$ have been studied in context of both divisible and indivisible goods \cite{kaneko1979nash,kaneko1980extension,nguyen2014minimizing,caragiannis2019unreasonable}.\footnote{Divisible goods correspond to items that can be fractionally divided among the agents and, complementarily, indivisible goods are ones that have to be integrally assigned.} Specifically, in the context of indivisible goods and additive valuations, \cite{caragiannis2019unreasonable} shows that any allocation that maximizes $\NSW$ is guaranteed to be envy-free up to one good. \\

\noindent
{\it Significance of Approximating Nash Social Welfare.} As mentioned previously, $\NSW$ stands on axiomatic foundations. In particular, the Pigou-Dalton principle ensures that $\NSW$ will increase by transferring, say, $\delta$ value from a well-off agent $i$ to an agent $j$ with lower current value. At the same time, if the relative increase in $j$'s value is much less than the drop experienced by agent $i$, then $\NSW$ will not favor such a transfer, i.e., this welfare function also accommodates for collective efficiency. From a welfarist perspective, $\NSW$ induces a cardinal ordering (ranking) among the allocations and a meaningful goal is to find an allocation with as high a Nash social welfare as possible. This viewpoint is standard in cardinal treatments: each agent prefers bundles with higher values, and the social planner prefers allocations (valuation profiles) with higher welfare (be it social, Nash, or egalitarian). Indeed, this objective pervades all welfare functions--approximating $\NSW$ (in fair division contexts) is as well motivated as approximating social welfare (when economic efficiency is of central concern). Furthermore, it is important to note that, while Nash optimal allocations might satisfy additional fairness guarantees, this fact does not undermine the relevance of finding allocations with as high an $\NSW$ as possible. Overall, computing allocations with high $\NSW$ is a well-justified objective in and of itself. These observations, in particular, motivate the study of approximation algorithms for $\NSW$ maximization.  \\

The current work addresses the problem of allocating \emph{indivisible} goods with the aim of maximizing Nash social welfare. We focus on fair division instances wherein the agents' valuations are $\XOS$ functions. Specifically, a set function $v$ is said to be $\XOS$ (fractionally subadditive) iff it is a pointwise maximizer of additive functions,\footnote{Recall that a set function $f$ is said to be additive iff, for all subsets $S$, the function value $f(S)$ is equal to the sum of the values of the elements in the set, $f(S) = \sum_{g \in S} f(\{g\})$.} i.e., iff there exists a family of additive functions $\mathcal{F}$ such that $v(S) = \max_{ f \in \mathcal{F}} f(S)$, for all subsets $S$. $\XOS$ functions constitute an encompassing class in the hierarchy of complement free valuations. This hierarchy has been extensively studied in the context of social welfare maximization \cite{AGTBook} and includes, in order of containment, the following valuation classes: additive, submodular, XOS, and subadditive. As detailed below, these function families have also been the focus of recent works on Nash social welfare maximization. 

\paragraph{Computational results for $\NSW$ maximization.} In the indivisible-goods context and for additive valuations, a series of notable works have developed constant-factor approximation algorithms for the $\NSW$ maximization problem. The formative work of Cole and Gkatzelis~\cite{cole2015approximating} obtained the first constant-factor approximation (specifically, $2e^{\sfrac{1}{e}}$) for maximizing $\NSW$ under additive valuations. With an improved analysis, an approximation ratio of $2$ for the problem was obtained in \cite{cole2017convex}. Also, an $e$-approximation has been achieved \cite{anari2016nash}; this result utilizes real stable polynomials. Currently, the best-known approximation ratio for additive valuations is $e^{1/e}$ \cite{barman2018finding}. 

Complementary to these positive results, the work of Garg et al.~\cite{DBLP:journals/corr/GargHM17} shows that, under additive valuations, it is {\rm NP}-hard to approximate $\NSW$ within a factor of $1.069$; see also \cite{lee2017apx}. Furthermore,  under submodular valuations, Garg et al.~\cite{garg2020approximating} showed that achieving an approximation ratio better than $e/(e-1)$ for $\NSW$ maximization---in the value-oracle model---is {\rm NP}-hard.  


In the context of additive-like valuations, a $(2.404 + \epsilon)$-approximation guarantee is known for budgeted additive valuations \cite{garg2018approximating} and a $2$-approximation has been obtained for separable piecewise linear concave (SPLC) valuations~\cite{anari2018nash}. 

For the broader class of submodular valuations, an $O(n \log n)$ approximation guarantee was achieved in \cite{garg2020approximating}; we will, throughout, use $n$ to denote the number of agents in the fair division instance. Furthermore, the recent work of Li and Vondr\'{a}k~\cite{vondraksubmodular} develops a constant-factor approximation algorithm for $\NSW$ maximization under submodular valuations. This result builds upon the work of Garg et al.~\cite{radogarg} that addresses Rado valuations.  

An $O(n)$-approximation ratio was obtained, independently, in \cite{Barman2020TightAA} and \cite{chaudhury2020fair} for the two most general valuation classes in above-mentioned hierarchy. That is, for $\NSW$ maximization a linear approximation guarantee can be achieved under $\XOS$ and subadditive valuations. Note that these algorithmic results hold under the standard value-oracle model, i.e., they only require values of different subsets (say, via an oracle). In fact, the work of Barman et al.~\cite{Barman2020TightAA} shows that, in the value-oracle model, this linear approximation guarantee is the best possible: under $\XOS$ (and, hence, subadditive) valuations, one necessarily requires exponentially many value queries to find any sublinear (in $n$) approximation for $\NSW$. This (unconditional) lower bound necessitates the use of stronger query models for breaking the $O(n)$-approximation barrier. Towards this, the current work utilizes demand oracles and $\XOS$ oracles. Both of these query models have been used in prior work on social welfare maximization under $\XOS$ and subadditive valuations \cite{dobzinski2006improved,feige2009maximizing,dobzinski2010approximation}. In particular, the work of Feige \cite{feige2009maximizing} uses demand oracles and achieves an $e/(e-1)$-approximation ratio for social welfare maximization under $\XOS$ valuations. 

Note that, for an $\XOS$ valuation $v$ defined (implicitly) by a family of additive functions $\mathcal{F}$, an $\XOS$ oracle, when queried with a subset $S$, returns a maximizing additive function $f \in \mathcal{F}$, i.e., the oracle returns $f \in \argmax_{f' \in \mathcal{F}} f'(S)$. Also, a demand oracle for valuation $v$ takes as input prices $p_g \in \mathbb{R}$, for all the goods $g$, and returns a set $S \subseteq [m]$ that maximizes $v(S) - \sum_{g\in S} p_g$.

\subsection{Our Results and Techniques}
We develop the first sublinear approximation algorithm for maximizing Nash social welfare under $\XOS$ valuations, specified via demand and $\XOS$ oracles. \\

\noindent
{\bf Main Result:} Given $\XOS$ and demand oracle access to the ($\XOS$) valuations of $n$ agents, one can compute in polynomial-time (and with high probability) an ${O}(n^{\sfrac{53}{54}})$ approximation for the Nash social welfare maximization problem. \\


Our algorithm (Algorithm \ref{Alg:NSWXOS} in Section \ref{sec:overview}) first finds a linear (in $n$) approximation using essentially half the goods (obtained via random selection). This linear guarantee is achieved using intricate extensions of the idea of repeated matchings \cite{garg2020approximating} and the discrete moving knife method \cite{Barman2020TightAA}. A key contribution of the work is to then develop a novel connection between $\NSW$ and social welfare under a capped version of the agents' valuations. In particular, we use the linear guarantees as benchmarks and define capped valuations for the agents. Subsequently, we partition the remaining goods to (approximately) maximize social welfare under these capped valuations. We show that (for a relevant subset of agents) these steps bootstrap the linear guarantee into a sublinear bound. Indeed, this connection between social welfare, under the capped valuations, and $\NSW$ might be of independent interest. We also note that maximizing social welfare under capped valuations (with oracle access to the agents' underlying valuations and not the capped ones) is an involved step in and of itself. We use multiple other techniques to overcome such hurdles and overall obtain a sublinear approximation ratio through a sophisticated analysis. Section \ref{sec:overview} provides a detailed overview of the algorithm and the main result is established in Section \ref{sec:analysis}. 

Furthermore, we complement, in part, the algorithmic result by showing that, under $\XOS$ valuations, an exponential number of demand and $\XOS$ queries are necessarily required to approximate $\NSW$ within a factor of $\left(1 - \frac{1}{e}\right)$; see Theorem \ref{thm:xoshard} in Section \ref{appendix:comm-comp-nsw}. This unconditional lower bound is obtained by establishing a communication complexity result: considering (for analysis) a setting wherein each agent holds her $\XOS$ valuation, we show that exponential communication among the agents is required for approximating $\NSW$ within a factor of  $\left(1 - \frac{1}{e}\right)$. Therefore, the query bound here holds not only for demand, $\XOS$, and value queries, but applies to any (per-valuation) query model in which the queries and the oracle responses are polynomially large.   

\subsection{Additional Related Work}
Nash social welfare maximization---specifically in the context of indivisible goods---has been prominently studied in recent years. Along with above-mentioned works, algorithmic results have also been developed for various special cases. In particular, the $\NSW$ maximization problem admits a polynomial-time (exact) algorithm for binary additive \cite{darmann2015maximizing,barman2018greedy} and binary submodular (i.e., matroid-rank) valuations \cite{babaioff}. Considering the particular case of binary $\XOS$ valuations, Barman and Verma \cite{barman2021approximating} show that a constant-factor approximation for $\NSW$ maximization can be efficiently computed in the value-oracle model. They also prove that, by contrast, for binary subadditive valuations, any sublinear approximation requires an exponential number of value queries. Nguyen and Rothe~\cite{nguyen2014minimizing} study instances with identical, additive valuations and develop a $\rm PTAS$ for maximizing $\NSW$ in such settings. 

In contrast to the above-mentioned results, the current work addresses the entire class of $\XOS$ valuations and obtains a nontrivial approximation guarantee. 

\section{Notation and Preliminaries}\label{sec:prelim}
We study the problem of discrete fair division, wherein $m$ indivisible goods have to be partitioned among $n$ agents in a fair manner. The cardinal preferences of the agents $i\in [n]$ (over a subset of goods) are specified via valuations $v_i\colon 2^{[m]}\mapsto \mathbb{R}_{+}$, where $v_i(S)\in \mathbb{R}_{+}$ is the value that agent $i$ has for subset of goods $S\subseteq [m]$. We denote an instance of a fair division problem by the tuple $\langle [n], [m], {\{v_i\}}_{i\in [n]} \rangle$.

This work focuses on $\XOS$ valuations. A set function $v \colon 2^{[m]} \mapsto \mathbb{R}_+$ is said to be $\XOS$ (or fractionally subadditive), iff there exists a family of additive set functions $\mathcal{F}$ such that, for each subset $S \subseteq[m]$, the value $v(S) =\max_{f \in \mathcal{F}} f(S)$. Note that the cardinality of the family $\mathcal{F}$ can be exponentially large in $m$. $\XOS$ valuations form a subclass of subadditive valuations; in particular, they satisfy $v(A\cup B) \leq v(A)+v(B)$, for all subsets $A,B \subseteq [m]$. We use $v_i(g)$ as a shorthand for $v_i(\{g\})$, i.e., for the value of good $g \in [m]$ for agent $i \in [n]$.

Since explicitly representing valuations (set functions) may require exponential space, prior works develop efficient algorithms assuming oracle access to the valuations. A basic oracle access is obtained through value queries: a value oracle, when queried with a subset $S$ returns the value of $S$. The current work uses demand oracles and $\XOS$ oracles. For an $\XOS$ valuation $v$ defined (implicitly) by a family of additive functions $\mathcal{F}$, an $\XOS$ oracle, when queried with a subset $S$, returns a maximizing additive function $f \in \mathcal{F}$, i.e., the oracle returns $f \in \argmax_{f' \in \mathcal{F}} f'(S)$. Note that such an additive function $f$ can be completely specified by listing the values $\{f(g)\}_{g \in [m]}$. Also, given that $v$ is $\XOS$ and $f \in \mathcal{F}$, we have  $v(T)\geq f(T)$, for all subsets $T\subseteq [m]$. 

A demand oracle for valuation $v$ takes as input a price vector $p=(p_1,p_2, \ldots, p_m) \in \mathbb{R}^m$ over the $m$ goods (i.e., a demand query) and returns a set $S \subseteq [m]$ that maximizes $v(S) - \sum_{g\in S} p_g$. It is known that a demand oracle can simulate a value oracle (via a polynomial number of demand queries), but the converse is not true  \cite{AGTBook}. 

We will, throughout, assume that the agents' valuations $v_i$s are normalized ($v_i(\emptyset) = 0$) and monotone: $v_i(A) \leq v_i(B)$ for all $A\subseteq B \subseteq [m]$. 

An allocation $\A=(A_1, A_2,\ldots, A_n)$ is an $n$-partition of the $m$ indivisible goods, wherein subset  $A_i$ is assigned to agent $i\in [n]$. Each such allocated subset $A_i \subseteq [m]$ will be referred to as a {bundle}. 

The goal of this work is to find allocations with as high a Nash social welfare as possible. Specifically, for a fair division instance $\langle [n],[m], \{v_i\}_{i \in [n]}\rangle$, the Nash social welfare, $\NSW(\cdot)$, of an allocation $\A$ is the geometric mean of the agents' valuations under $\A$, i.e., $\NSW(\A) \coloneqq \left( \prod_{i=1}^n v_i(A_i) \right)^{\sfrac{1}{n}}$. Throughout, we will write $\mathcal{N}=(N_1,\ldots, \, N_n)$ to denote an allocation that maximizes the Nash social welfare in the given instance and will refer to such an allocation as a Nash optimal allocation. In addition, let $g^*_i$ denote the good most valued  by agent $i$ in the bundle $N_i$, i.e., $g^*_i \in \argmax_{g \in N_i} v_i(g)$. We will assume, throughout, that $\NSW(\mathcal{N}) >0$ and, hence, $v_i(g^*_i) >0$. For the complementary case, wherein $\NSW(\mathcal{N})=0$, returning an arbitrary allocation suffices. With parameter  $\alpha \geq 1$, an allocation $\A=(A_1,\ldots, \, A_n)$ is said to be an $\alpha$-approximate solution for the problem of maximizing Nash social welfare iff $\NSW(\A) \geq \frac{1}{\alpha}\NSW(\mathcal{N})$.

For subsets $S, T\subseteq [m]$, we will use the shorthand $S + T  \coloneqq S \cup T$ and $S - T \coloneqq S \setminus T$. Furthermore, for good $g \in [m]$, we will write $S + g$ to denote $S + \{ g \}$.

\section{Algorithm and Technical Overview}\label{sec:overview} 
\floatname{algorithm}{Algorithm}
\begin{algorithm}[h!]
    \caption{Sublinear approximation for Nash social welfare under $\XOS$ valuations}\label{Alg:NSWXOS} 
    \textbf{Input:} Instance $\instance{[n], [m], \val}$ with demand and $\XOS$ oracle access to the ($\XOS$) valuations $v_i$s \\
    \textbf{Output:} Allocation $\Q=(Q_1,\ldots, Q_n)$  
    \begin{algorithmic}[1]
        \STATE Initialize $\M = \emptyset$ and $\G = [m]$ \label{step:initialize}
        \FOR {$t = 1$ to $\log n$}\label{step:M}
        \STATE Find matching $\tau_t: [n] \rightarrow \G$ that maximizes $\prod_{i \in [n]} v_i \left(\tau_t(i) \right)$ \label{step:tau}
        \STATE Update $\G \leftarrow \G - \left\{ \tau_t(i) \right\}_{i \in [n]}$ and $\M \leftarrow \M + \left\{ \tau_t(i) \right\}_{i \in [n]}$
        \ENDFOR \label{step:endfor}
        \STATE Find a matching $\pi : [n] \rightarrow \G$ that maximizes $\prod_{i \in [n]} v_i(\pi(i))$ \label{step:pi}
        \STATE Update $\G \leftarrow \G - \{ \pi(i) \}_{i \in [n]}$ 
        
        \STATE Randomly partition the set of goods $\G$ into $\R$ and $\R'$, i.e., each good in $\G$ is included in $\R$, or $\R'$, independently with probability $\nicefrac{1}{2}$. \\
        
        \STATE Set allocation $(X_1, X_2, \ldots, X_n) = {\textsc{DiscreteMovingKnife}} \left([n], {\R}, \{v_i\}_{i \in [n]} \right)$\label{step:MK} \\
        \COMMENT{This subroutine is detailed in Section \ref{section:DMK}} \\
        
        \STATE For each $i\in [n]$, set (scaling factor) $\beta_i \coloneqq \frac{1}{n} \cdot \frac{1}{v_i(X_i + \pi(i) )}$ \label{step:beta}
        \STATE Set allocation $(Y_1, Y_2, \ldots, Y_n) =  {\textsc{CappedSocialWelfare}} \left( [n], {\R}', \{v_i\}_{i \in [n]}, (\beta_i)_{i \in [n]} \right)$ \label{step:SW} \\
        \COMMENT{This subroutine is detailed in Section \ref{sec:capsw}} \\
        
        \STATE Find matching $\mu: [n] \rightarrow \M$ that maximizes $\prod_{i \in [n]} v_i \left( \mu(i) + \pi(i) + X_i + Y_i \right)$\label{step:mu} \COMMENT{Note that $\mu$ assigns to each agent a good from the set $\M$, which was populated in the for-loop.}
        \RETURN allocation $\left(Q_i = \mu(i) + \pi(i) + X_i + Y_i \right)_{i \in [n]}$
    \end{algorithmic}
\end{algorithm}
Our main algorithm (Algorithm $\ref{Alg:NSWXOS}$) consists of the following four phases:   
\begin{itemize}
\item[({\rm I})] Keep aside a set of high-valued goods $\M$ (via the for-loop in Steps \ref{step:M} to \ref{step:endfor}) and allocate $n$ goods via a matching, $\pi$ between set of agents $[n]$ and leftover goods, $[m] \setminus \M$ (in Step \ref{step:pi}). 
\item[({\rm II})] Randomly partition the remaining goods into two parts, $\R$ and $\R'$. 
\item[({\rm III})] Allocate the subset of goods $\R$ among the agents---as bundles $X_1, X_2, \ldots, X_n$---via a discrete moving knife procedure. 
\item[(\rm{IV})] From the goods in $\R'$, find an allocation $(Y_1, \ldots, Y_n)$ that (approximately) maximizes social welfare under (judiciously defined) \emph{capped valuations}. 
\end{itemize}

\paragraph{Phase {\rm I}.} In the first phase, the algorithm identifies a subset of goods $\M$ that are kept aside while executing the intermediate phases. The algorithm rematches within $\M$ before termination (Step \ref{step:mu}). The algorithm populates the set $\M$ by repeatedly finding matchings: it initializes $\M = \emptyset$, $\G = [m]$, and considers the complete weighted bipartite graph between the set of agents $[n]$ and the set of goods $\G$. The edge  between agent $i$ and good $g$ has weight $w(i,g)=\log v_i(g)$. The algorithm then computes a maximum-weight matching $\tau$ in this bipartite graph and includes the matched goods $\{ \tau(i) \}_{i \in [n]}$ in $\M$ (i.e., $\M \leftarrow \M + \{ \tau(i) \}_{i \in [n]}$). Removing the matched goods from consideration, $\G \leftarrow \G - \{ \tau(i) \}_{i \in [n]}$, the algorithm repeats this procedure $\log n$ times. 

We will show that the set $\M$ contains, for each agent $i$, a distinct good of value at least $v_i(g^*_i)$ (see Lemma \ref{lemma:matchhigh}); recall that $g^*_i$ is the most-valued good in agent $i$'s optimal bundle, $g^*_i \in \argmax_{g \in N_i} v_i(g)$. At a high level, this property will be used to establish approximation guarantees for agents $i$ that receive sufficiently high value via just the single good $g_i^*$. 

The algorithm also seeds the allocation of each agent by finding a matching $\pi$ (in Step \ref{step:pi})---between $[n]$ and the goods $[m] \setminus \M$---that maximizes the product (equivalently, the geometric mean) of the valuations. Each agent $i$ is permanently assigned the good $\pi(i)$.

Here, if the product of the values is zero ($\prod_{i=1}^n v_i(\pi(i)) = 0$), we consider matchings that maximize the number of agents who achieve a nonzero value (i.e., consider maximum-cardinality matchings with nonzero values) and among them select the one that maximizes the product.\footnote{Such a matching can be computed efficiently as a maximum-cardinality maximum-weight matching in the agents-goods bipartite graph; here the edge weight between agent $i$ and good $g$ is set to be $\log v_i(g)$, for nonzero $v_i(g)$s.} Furthermore, all agents $z \in [n]$ with $v_z(\pi(z)) = 0$ will be excluded from consideration till Step \ref{step:mu} of the algorithm, i.e., such agents $z$ will not participate in phases three and four. For ease of presentation, we assume that there are no such agents (i.e., $v_i(\pi(i))>0$ for all $i$). This assumption does not affect the approximation guarantee; see the remark at the end of Section \ref{subsection:mainproof}. Furthermore, the assumption ensures that parameters $\beta_i$s considered in the fourth phase (Step \ref{step:beta}) are well-defined.

\paragraph{Phase {\rm II}.} The remaining goods $\mathcal{G} = [m] \setminus \left( \M + \pi([n]) \right)$ are partitioned randomly into two subsets $\R$ and $\R'$. Intuitively, the first phase addresses agents $i$ for whom good $g^*_i \in N_i$ by itself achieves a sublinear guarantee. Complementarily, the random partitioning and the subsequent phases are essentially aimed at agents $j$ for whom all the goods in $N_j$ are of sufficiently small value. This small-valued goods property (specifically, $v_j(g) \leq \frac{1}{\sqrt{n}} v_j(N_j)$, for all $g \in N_j$) ensures that, with high probability and for relevant agents $j$, both the values $v_j(N_j \cap \R)$ and $v_j(N_j \cap \R')$ are within a constant factor of $v_j(N_j)$. We prove this via concentration bounds (see Lemma \ref{lemma:concentration}). Phases {\rm III} and {\rm IV} utilize the subsets of goods $\R$ and $\R'$, respectively. The fact that, for relevant agents $j$, a near-optimal bundle exists in both $\R$ and $\R'$ enables us to (a) obtain a linear approximation for the concerned agents in Phases {\rm III} and (b) bootstrap the linear guarantee to a sublinear one in Phase  {\rm IV}. Notably, the current bootstrapping method goes beyond the substantial collection of techniques that have been recently developed for $\NSW$ maximization and it might be applicable in other resource-allocation contexts.  



\paragraph{Phase {\rm III}.} This phase partitions the subset of goods $\R$. As mentioned above, our aim here is to obtain a linear approximation for an appropriate set of agents. Towards that, we consider agents $i$ with the property that $v_i(g) \lesssim \frac{1}{n} v_i(\widetilde{N}_i)$, for a near-optimal bundle $\widetilde{N}_i$ and all goods $g \in \widetilde{N}_i$; see Section \ref{section:DMK} for details. Now, to achieve a linear approximation guarantee for such agents $i$, we develop a discrete moving knife subroutine (Algorithm \ref{Alg:MovingKnife} in Section \ref{section:DMK}). Moving knife methods, in general, start with agent-specific (value) thresholds and then iteratively assign bundles (to the agents) that satisfy these thresholds. In the current context and to address the relevant subset of agents, we first restrict the valuation of each agent $j$ to the subset of goods that are individually of small value (for $j$) and then execute the moving knife method. This modification ensures that the developed subroutine finds an allocation $(X_1, \ldots, X_n)$ with the desired linear approximation guarantee. 


\paragraph{Phase {\rm IV}.} A distinguishing idea in the current work is to use $v_i(X_i)$s (which provide a linear approximation for a relevant subset of agents) as a benchmark and bootstrap towards the desired sublinear bound in this phase. In particular, we use the values achieved in the allocation computed in Phase {\rm III}---i.e., in $(X_1, \ldots, X_n)$---to define, for each agent $i$, a scaling factor $\beta_i \coloneqq \frac{1}{n} \cdot \frac{1}{v_i(X_i + \pi(i) )}$ (Step \ref{step:beta}). Furthermore, we consider (capped) valuation $ \widehat{v}_i(T) \coloneqq \min \left\{ \frac{1}{\sqrt{n}}, \ \beta_i v_i(T) \right\}$, for all subsets $T \subseteq \R'$. 

The algorithm partitions the remaining subset of goods $\R'$ in Step \ref{step:SW} by executing the subroutine \textsc{CappedSocialWelfare}. The objective of the subroutine is to (approximately) maximize social welfare under $\widehat{v}_i$s. Intuitively, the parameters $\beta_j$s are set to ensure that, when maximizing social welfare under $\widehat{v}_j$s, one prefers agents $i$ for whom $v_i(X_i + \pi(i))$ is much smaller than their optimal value $v_i(N_i)$. In Section \ref{sec:capsw}, we detail the \textsc{CappedSocialWelfare} subroutine and show that its computed allocation $(Y_1, \ldots, Y_n)$ bootstraps the approximation guarantee towards the desired sublinear bound. 

Note that maximizing social welfare is usually not aligned with the goal of maximizing $\NSW$; an allocation with high social welfare can leave a small subset of agents with zero value and, hence, such an allocation would have zero Nash social welfare. Hence, approximating Nash social welfare by approximately maximizing social welfare (under capped valuations) is an interesting connection. This connection builds on the intricate guarantees obtained in the other phases and, in particular, it entails: $(i)$ carefully defining the capped valuations functions $\widehat{v}_i$s such that (a non-trivial fraction of) agents who have received low value in previous phases receive a high value in this phase, and $(ii)$ maximizing social welfare under $\widehat{v}_i$s, with oracle access to the underlying valuations $v_i$s. Here, requirement $(ii)$ is nontrivial, since oracles for the valuations, $v_i$s, need to be appropriately modified to address $\widehat{v}_i$s. Such a modification is simple for value oracles, but not for demand and $\XOS$ oracles. We develop subroutine \textsc{CappedSocialWelfare} (detailed in Section \ref{sec:capsw}) that overcomes these challenges and returns the desired allocation $(Y_1, \ldots, Y_n)$. 

After these four phases, the algorithm finds a matching $\mu$ into the set of goods $\M$, which was initially kept aside. The matching $\mu$ maximizes the product of the valuations with offset $X_i+Y_i+ \pi(i)$ (Step \ref{step:mu}).\footnote{Note that the matchings in Steps \ref{step:tau}, \ref{step:pi}, and \ref{step:mu} can be efficiently computed by finding a maximum-weight matching in a bipartite graph between the agents and the relevant goods;  here, the weight of each edge is set as the $\log$ of the appropriate value.} Finally, each agent $i \in [n]$ is assigned the bundle $Q_i \coloneqq \pi(i) + X_i + Y_i + \mu(i)$.

In Section \ref{sec:phases}, we detail the above-mentioned phases and establish relevant guarantees for each. The phases work in close conjunction with each other; detailed guarantees from earlier phases support the successful executions of the latter ones. Overall, an intricate analysis is required to obtain the desired sublinear approximation guarantee. We accomplish this in Section \ref{sec:analysis} and establish the approximation ratio for the returned allocation $\mathcal{Q} = (Q_1, \ldots, Q_n)$. Specifically, we prove the following theorem in Section \ref{sec:analysis}.


\begin{restatable}{theorem}{mainthm}({\bf Main Result.}) 
\label{thm:main}
Given a fair division instance, $\instance{[n], [m], \{v_i\}_{i \in [n]}}$, with $\XOS$ and demand oracle access to the (monotone and $\XOS$) valuations $v_i$s, Algorithm $\ref{Alg:NSWXOS}$ computes (with high probability) an ${O}(n^{\sfrac{53}{54}})$ approximation to the optimal Nash social welfare.
\end{restatable}

Finally, in Section \ref{appendix:comm-comp-nsw}, we complement our algorithmic result, in part, with a query complexity result. We prove that  exponential communication is required to approximate $\NSW$ within a factor of $(1 - 1/e)$. This lower bound on communication complexity directly provides a commensurate lower bound under the considered query models, i.e., under value, demand, and $\XOS$ queries. To prove the negative result we reduce the problem of $\textsc{MultiDisjointness}$ to that of maximizing Nash social welfare. $\textsc{MultiDisjointness}$ is a well-studied problem in the communication complexity literature. In this problem we have $n$ players and each player $i \in [n]$ holds a subset $B_i$ of a ground set of elements $[t]$. It is known that distinguishing between the cases of totally intersecting (i.e., there is an element that is included in $B_i$ for all $i \in [n]$) and totally disjoint (i.e., $B_i \cap B_j = \emptyset$ for all $i \neq j$) requires $\Omega \left(\sfrac{t}{n} \right)$ communication. We reduce this problem to $\NSW$ maximization with each agent $i$ holding her $\XOS$ valuation $v_i$; in the reduction $t$ dictates the number of additive functions that define the $\XOS$ valuation of each agent. The key idea here is to show that there exists $\XOS$ valuations such that in the totally intersecting case the optimal $\NSW$ is sufficiently high and, complementarily, in the totally disjoint case it is sufficiently low. That is, between the two underlying cases, the optimal $\NSW$ bears a multiplicative gap of at least $\left(1 - 1/e\right)$. Therefore, the reduction shows that approximating $\NSW$ within a factor of $(1 - 1/e)$ necessarily requires $\Omega \left(\sfrac{t}{n} \right)$ communication. With an exponentially large $t$, we obtain the desired query lower bound. Formally, we establish the following theorem 


\begin{restatable}{theorem}{xoshard}\label{thm:xoshard}
For fair division instances with $\XOS$ valuations and a fixed constant $\varepsilon \in (0,1]$, exponentially many demand and $\XOS$ queries are necessarily required for finding an allocation with $\NSW$ at least $\left(1 - \frac{1}{e} + \varepsilon\right)$ times the optimal.  
\end{restatable}

\section{Phases of Algorithm \ref{Alg:NSWXOS}} \label{sec:phases}

\subsection{Phase {\rm I}: Isolating High-Valued Goods via Repeated Matching}\label{subsec:match}

Recall that $\M$ denotes the set of goods identified in Phase {\rm I} of Algorithm \ref{Alg:NSWXOS} (see the for-loop at Step \ref{step:M}). Also, $\mathcal{N} = (N_1, \ldots, N_n)$ denotes a Nash optimal allocation and $g^*_i \in \argmax_{g \in N_i} v_i(g)$, for all agents $i \in [n]$. As mentioned previously, Algorithm \ref{Alg:NSWXOS} keeps the goods in $\M$ aside while executing intermediate phases and at the end rematches within $\M$ (Step \ref{step:mu}). 

The following lemma shows that $\M$ admits a matching $h$ wherein each agent receives a good with value at least that of $g^*_i$. At a high level, this lemma will be used to establish an approximation guarantee for agents $i$ that receive sufficiently high value via just the single good $g_i^*$.  The proof of this lemma is deferred to Appendix \ref{appendix:repeated-mathcings}.

\begin{restatable}{lemma}{LemmaRepeatedMatchings}
\label{lemma:matchhigh}
There exists a matching $h \colon [n] \mapsto \M$ such that $v_i(h(i))\geq v_i(g^*_i)$, for all agents $i\in [n]$. 
\end{restatable}

\subsection{Phase {\rm II}: Randomly Partitioning Goods} \label{subsec:conc}
The following lemma addresses (near-optimal) bundles $\overline{N}_i \subseteq [m]$ with no high-valued goods. For such bundles, the lemma shows that randomly partitioning the goods---into two subsets $\R$ and $\R'$---preserves, with high probability, sufficient value of $\overline{N}_i$ in both the parts, i.e., both $v_i(\oveN_i \cap \R)$ and $v_i(\oveN_i \cap \R')$ are comparable to $v_i(\oveN_i)$. Hence, for bundles with no high-valued goods, the lemma implies that one can obtain sufficiently high welfare (Nash and social) in $\R$ as well as $\R'$. The proof of the lemma appears in Appendix \ref{appendix:concentration}.

\begin{restatable}{lemma}{LemmaConcentration}
\label{lemma:concentration}
Let $\G$ be a set of indivisible goods, $v_i$ be the $\XOS$ valuation of an agent $i \in [n]$, and $\overline{N}_i \subseteq \G$ be a subset with the property that $\max_{g \in \overline{N}_i} v_i(g) \leq \frac{1}{\sqrt{n}} v_i(\overline{N}_i)$. Then, for a random partition of $\G$ into sets $\R$ and $\R'$, we have 
\begin{align*}
    \Pr \left\{ v_i(\oveN_i \cap \R) \leq \frac{1}{3} v_i(\oveN_i) \right\} &\leq  \exp \left( -\frac{\sqrt{n}}{18}\right) \qquad \text{ and }\\
    \Pr \left\{ v_i(\oveN_i \cap \R') \leq \frac{1}{3} v_i(\oveN_i) \right\}  &\leq \exp \left( -\frac{\sqrt{n}}{18}\right).
\end{align*}
Here, random subset  $\R \subseteq \G$ is obtained by selecting each good in $\G$ independently with probability $\nicefrac{1}{2}$, and $\R' \coloneqq \G - \R$.
\end{restatable}

Applying union bound, we extend Lemma \ref{lemma:concentration} to obtain the following result for allocations $(\oveN_1,\ldots, \oveN_n)$.

\begin{lemma}\label{lem:concbound}
Given a set of indivisible goods $\G$ along with $\XOS$ valuations $v_i$ for agents $i \in [n]$, and a partition $(\oveN_1,\ldots, \oveN_n)$ of $\G$, let subset $T \coloneqq \{i\in [n] \,:\, \max_{g\in \oveN_i} v_i(g)\leq \frac{1}{\sqrt{n}} v_i(\oveN_i)\}$. Then, for a random partition of $\G$ into sets $\R$ and $\R'$, we have 
\begin{align*}
	\Pr\left\{ v_i(\oveN_i \cap \R) \geq \frac{1}{3} v_i(\oveN_i) \text{, for all } i \in T \right\} &\geq 1- n \exp{ \left( -\frac{\sqrt{n}}{18}\right)} \qquad \text{ and } \\
    \Pr\left\{ v_i(\oveN_i \cap \R') \geq \frac{1}{3} v_i(\oveN_i) \text{, for all } i \in T \right\} &\geq 1- n \exp{ \left( -\frac{\sqrt{n}}{18}\right)}.
\end{align*}
Here, random subset  $\R \subseteq \G$ is obtained by selecting each good in $\G$ independently with probability $\nicefrac{1}{2}$, and $\R' \coloneqq \G - \R$.
\end{lemma}

\subsection{Phase {\rm III}: Discrete Moving Knife}\label{subsec:mknife}
\label{section:DMK}
This section presents the \textsc{DiscreteMovingKnife} subroutine (Algorithm \ref{Alg:MovingKnife}). As mentioned previously, the subroutine is designed to address agents $i \in [n]$ for whom there exists (near-optimal) bundles $\widetilde{N}_i$ with the property that $\max_{g \in \widetilde{N}_i} \ v_i(g) \leq \frac{1}{16 n} v_i(\widetilde{N}_i)$. Specifically, the subroutine obtains a linear approximation with respect to any allocation $(\widetilde{N}_1, \ldots, \widetilde{N}_n)$ and for the corresponding set of agents $T \coloneqq \{ j \in [n] \ : \  \max_{g \in \widetilde{N}_j} \ v_j(g) \leq \frac{1}{16 n} v_j(\widetilde{N}_j)\}$; see Lemma \ref{lem:MK} below. Indeed, the \textsc{DiscreteMovingKnife} subroutine does not explicitly require as input an (near-optimal) allocation $(\widetilde{N}_1, \ldots, \widetilde{N}_n)$.

Given a set of indivisible goods $\R$ to partition, the subroutine (Algorithm \ref{Alg:MovingKnife}) first finds, for each agent $j \in [n]$, a subset of goods $G_j \subseteq \R$ that solely consists of small-valued goods, i.e., $G_j$ satisfies $v_j(g) < \frac{1}{16 n} v_j(G_j)$ for all $g \in G_j$. The set $G_j$ is computed by iteratively removing goods that violate the small-value requirement (see the while-loop in Step \ref{step:stopcondition} of Algorithm \ref{Alg:MovingKnife}). Observe that, by construction, for each $j$, when the while-loop (Step \ref{step:endwhileG_i}) terminates the set $G_j$ satisfies $v_j(g) < \frac{1}{16 n} v_j(G_j)$ for all $g \in G_j$. Also, note that, as $G_j$ shrinks in the while-loop, the value $v_j(G_j)$ decreases. However we show that, for any allocation $(\widetilde{N}_1, \ldots, \widetilde{N}_n)$ and any agent $i$ from the set $T \coloneqq \{j \in [n] \ : \  \max_{g \in \widetilde{N}_j} \ v_j(g) \leq \frac{1}{16 n} v_j(\widetilde{N}_j)\}$, the computed $G_i$ still satisfies $G_i \supseteq \widetilde{N}_i$  (see Claim \ref{claim:Gsupset} below). That is, to obtain a linear approximation for agents in $T$, it suffices to find an allocation $(X_1, \ldots, X_n)$ with the property that $v_i(X_i) \geq \frac{1}{16 n} v_i(G_i)$ for all $i \in T$. The subsequent steps of the \textsc{DiscreteMovingKnife} subroutine find such an allocation. In particular, for each agent $j \in [n]$, the subroutine restricts attention to the set $G_j$, i.e., considers valuation $v'_j(S)\coloneqq v_j(S\cap G_j)$, for all subsets $S\subseteq \R$ (Step \ref{step:definevprime}). This construction ensures that for each agent $j \in [n]$ and all goods $g \in \R$ we have $v'_j(g) \leq \frac{1}{16n} v'_j(G_j ) = v'_j(\R)$.

The subroutine then goes over all the goods in $\R$ in an arbitrary order and adds them one by one into a bundle $P$, until an agent $a$ calls out that her value (under $v'_a$) for $P$ is at least $\frac{1}{16 n} v'_a(\R)$. We assign these goods to agent $a$ and remove them (along with agent $a$) from consideration (Step \ref{step:DMKUpdate}). The subroutine iterates over the remaining set of agents and goods. Note that the subroutine only requires value-oracle access to the valuations $v_j$s.\footnote{We can efficiently simulate the value oracle for $v'_j$ as follows: for any queried subset $S$, the value $v'_j(S)$ can be obtained by querying for $v_j(G_j \cap S)$.}     

The following lemma (proved in Appendix \ref{appendix:DMK}) shows that the computed allocation $(X_1, X_2, \ldots, X_{n})$ achieves the desired linear approximation.

\floatname{algorithm}{Algorithm}
\begin{algorithm}[ht]
	\caption{\textsc{DiscreteMovingKnife} } \label{Alg:MovingKnife}
	\textbf{Input:} Instance $\langle [n], \R, \val \rangle$ with value-oracle access to the valuations $v_i$s \\
	\textbf{Output:} An allocation $(X_1, X_2, \ldots, X_{n})$  
	\begin{algorithmic} [1]
		\STATE For each agent $j \in [n]$, initialize set $G_j = \R$  \label{step:threshset}
		\FOR{$j = 1$ to $n$}
		\WHILE{there exists a good ${g} \in G_j$ such that $v_{j} (g) \geq \frac{1}{\constMK n} v_{j} (G_{j})$}\label{step:stopcondition}
		\STATE Update $ G_{j} \gets G_{j} - \{ {g} \}$
		\ENDWHILE\label{step:endwhileG_i}
		\ENDFOR
		\STATE Define $v'_j(S)\coloneqq v_j(S\cap G_j)$ for all agents $j\in [n]$ and subsets $S\subseteq \R$ \label{step:definevprime}
		\STATE Initialize set of goods $\Gamma=\R$ along with agents $A = [n]$ and bundles $X_j =\emptyset$, for all $j \in [n]$. Also, set $P = \emptyset$. 
		\WHILE {$\Gamma \neq \emptyset $ and $ A \neq \emptyset $}\label{step:loopESAMK}
		\STATE Pick an arbitrary good ${g} \in \Gamma$, and update $P\leftarrow P + \{{g} \}$ along with $ \Gamma \leftarrow  \Gamma - \{ {g} \}$
		\IF{there exists an agent $a \in A$ such that $v'_{a} (P) \geq \frac{1}{\constMK n} v'_{a} (\R)$ } \label{step:ifcondt}
		\STATE Assign $X_{a} = P$ and update $A \leftarrow A - \{ a \}$ along with $P = \emptyset$ \label{step:DMKUpdate}
		\ENDIF
		\ENDWHILE	
		\STATE If $\Gamma \neq \emptyset$, update $X_{n} \leftarrow X_{n} + \Gamma$
		\RETURN allocation $(X_1, X_2, \ldots, X_{n}).$
	\end{algorithmic}
\end{algorithm}

\begin{restatable}{lemma}{LemmaDiscreteMovingKnife}
\label{lem:MK}
Let $\langle [n], \R , \{v_i\}_{i\in [n]} \rangle$ be a fair division instance with $\XOS$ valuations $v_i$s. Also, let $(\widetilde{N}_1, \ldots, \widetilde{N}_n)$ be any allocation with $T \coloneqq \left\{ i \in [n]  \ : \ \max_{g \in \widetilde{N}_i} \ v_i(g) < \frac{1}{16n} v_i(\widetilde{N}_i) \right\}$. Then, given value-oracle access to $v_i$s, the  \textsc{DiscreteMovingKnife} subroutine computes---in polynomial time---an allocation $(X_1,\ldots,X_{n})$ with the property that $v_i(X_i)\geq \frac{1}{\constMK n} v_i(\widetilde{N}_i)$ for all $i\in T$.
\end{restatable}

\subsection{Phase IV : Maximizing Capped Social Welfare}\label{sec:capsw}
The section presents the \textsc{CappedSocialWelfare} subroutine (Algorithm \ref{Alg:CapSW}) that maximizes social welfare under capped versions of the  given valuations $v_i$s. Specifically, given a fair division instance $\instance{[n], {\R'}, \{v_i\}_{i \in [n]}}$ and parameters $\beta_1, \ldots, \beta_n  \in \mathbb{R}_+$, we define capped valuations, for each agent $i \in [n]$, as follows\footnote{Algorithm \ref{Alg:NSWXOS} invokes the subroutine \textsc{CappedSocialWelfare} with $\beta_i = \frac{1}{n}  \frac{1}{v_i(X_i + \pi(i) )}$. However, the results obtained in this section hold for any positive $\beta_i$s.} 
\begin{align}
    \widehat{v}_i(S) \coloneqq \min \left\{ \frac{1}{\sqrt{n}}, \  \beta_i  v_i(S) \right\} \qquad \text{ for all subsets $S$} \label{eqn:hat-v}
\end{align}
Since the valuations $v_i$s are $\XOS$, the functions $\widehat{v}_i$s are subadditive. Also, note that, using the value oracle for $v_i$, we can easily implement the value oracle for $\widehat{v}_i$. However, a key hurdle for the subroutine is that it does not have demand oracles for $\widehat{v}_i$s; otherwise, one could directly invoke the approximation algorithm of Feige \cite{feige2009maximizing} to maximize social welfare. We design the subroutine to overcome this hurdle and (approximately) maximize the social welfare under $\widehat{v}_i$s, using ($\XOS$ and demand) oracle access to $v_i$s.  

Our approximation guarantee (for social welfare under capped valuations) holds for instances $\instance{[n], {\R'}, \{v_i\}_{i}}$ wherein there exists an allocation $(O_1, \ldots, O_n)$ and a subset of agents $\overline{A} \subseteq [n]$ that satisfy   
\begin{itemize}
    \item[] \textbf{P1}: The welfare $\sum_{i\in \overline{A}} \ \widehat{v}_i(O_i) \geq \frac{26}{27} \sqrt{n}$.
    \item[] \textbf{P2}: For each agent $i\in \overline{A}$ and all goods $g'\in O_i$, the value $\widehat{v}_i(g')\leq \frac{1}{2\sqrt{n}}$.
\end{itemize}
When, in the analysis of the main algorithm (i.e., in Section \ref{sec:analysis}), we invoke the guarantee obtained here we will show that these two properties hold for the instance at hand. Also, note that both the properties express conditions in terms of the capped valuations $\widehat{v}_i$s. In particular, Property \textbf{P2} states that, for each agent $i$ in the designated set $\overline{A}$, all the goods in the bundle $O_i$ are of sufficiently small value. 
Property \textbf{P1} demands that we have high enough welfare (under $\widehat{v}_i$) among the bundles $O_i$ assigned to agents $i \in \overline{A}$.  

For $\XOS$ valuation $v_i$, let $\mathcal{F}_i$ denote the family of additive functions that define $v_i$. Throughout this section, we will write $f_{i, S}$ to denote the additive function in $\mathcal{F}_i$ that induces $v_i(S)$, i.e., for any subset $S$, 
\begin{align}
f_{i, S} \coloneqq \argmax_{f \in \mathcal{F}_i} f(S) \label{defn:fiS}
\end{align}

\floatname{algorithm}{Algorithm}
\begin{algorithm}[ht]
    \caption{\textsc{CappedSocialWelfare} } \label{Alg:CapSW}
    \textbf{Input:} Instance $\I = \langle [n], \R', \{v_i \}_{i \in [n]} \rangle$, with demand and $\XOS$ oracle access to the valuations $v_i$s, and parameters $\{\beta_i\}_{i \in [n]}$ \\
    \textbf{Output:} Allocation $(Y_1, \ldots, Y_n)$ 
    \begin{algorithmic} [1]
        \STATE Initialize $Y_i = \emptyset$, for all $i\in [n]$, and $Y_0 = \R'$  \COMMENT{$Y_0$ is the set of unallocated goods}
        \STATE $\rm{Flag} \gets \TRUE$
        \WHILE {Flag}\label{step:loopCapwelfare}
        \STATE For every good $g \in Y_0$ set price $p_g = 0$ \label{step:pricez}
        \STATE \label{step:price} For every agent $j \in [n]$ and bundle $Y_j$, query the $\XOS$ oracle for $v_j$ to find additive function $f_{j, Y_j}(\cdot)$. For each $g \in Y_j$, set price $p_g = 2\beta_j \ f_{j,Y_j}(g)$  \\
        \label{step:initpg}
        
        \FOR{each agent $j \in [n]$}
            \STATE For each good $g \in \R'$ with $\widehat{v}_j(g) \leq \frac{1}{2\sqrt{n}}$, set price $q_g^j=p_g$ \label{step:qprice1}
	        \STATE For each good $g \in \R'$ with $\widehat{v}_j(g) > \frac{1}{2\sqrt{n}}$, set price $q_g^j=\infty$ \label{step:qprice2}
    	    \STATE Let $D_j =\{g_1,\ldots, g_{|D_j|}\}$ be the demand set under valuation $v_j$ and prices $\sfrac{q_g^j}{\beta_j}$ \label{step:demand} \\ \COMMENT{Set $D_j$ is obtained via the given demand oracle for $v_j$. The goods in this set, $g_1, \ldots, g_{|D_j|}$, are indexed in an arbitrary order.}
	        \STATE Let $k$ be the minimum index such that $\widehat{v}_j(\{g_1,\ldots g_k\})\geq \frac{92}{225} \frac{1}{\sqrt{n}}$ \label{step:indexk}\\
	        \STATE Set $\widehat{D}_j=\{g_1,\ldots, g_k\}$ \label{step:Dhat} \\ \COMMENT{In case $\widehat{v}_j(D_j) < \frac{92}{225\sqrt{n}}$, set $\widehat{D}_j = D_j$}
        \ENDFOR
        \IF{there exists an agent $a \in [n]$ such that $\widehat{v}_a(\widehat{D}_a) +\sum_{j \in [n] \setminus\{ a \}} \ \widehat{v}_j(Y_j -\widehat{D}_a) \geq \sum_{j=1}^n \widehat{v}_j(Y_j) \ + \ \frac{1}{225\sqrt{n}}$} \label{step:ifcondt}
            \STATE Assign  $Y_a = \widehat{D}_a$ \label{step:iUpdate}
	        \STATE For all $j \in [n] \setminus \{a\}$, update $Y_j\gets Y_j - \widehat{D}_a$. \label{step:reassign}
	        \STATE Also, update the set of unallocated goods $Y_0 = \R' \setminus \left( \cup_{j=1}^n Y_j \right)$ 
        \ELSE
            \STATE $\rm{Flag} \gets \FALSE$ 
        \ENDIF
        \ENDWHILE
       
        \RETURN Allocation $(Y_1,\ldots ,Y_{n}).$
       
    \end{algorithmic}
\end{algorithm}

The subroutine \textsc{CappedSocialWelfare} (Algorithm \ref{Alg:CapSW}) starts with empty bundles, $Y_i = \emptyset$ for all agents $i \in [n]$, and with the set of unallocated goods $Y_0 = \R'$. Throughout, $Y_0$ denotes the set of unallocated goods with the maintained allocations $(Y_1, \ldots, Y_n)$. The subroutine iteratively transfers goods from $Y_0$ and between bundles as long as an increase in social welfare (with respect to $\widehat{v}_i$s) is obtained. Throughout its execution, the subroutine considers the efficacy of transferring a subset of goods $\widehat{D}_a$, to agent $a$, based on the current (social welfare) contribution of each good $g \in \widehat{D}_a$. In particular, for each agent $j$, we consider the contribution of the goods $g \in Y_j$ with respect to the additive function that induces $v_j(Y_j)$. Therefore, the sum of these contributions over $g \in Y_j$ is equal to $v_j(Y_j)$. For every good $g$ we set the price $p_g$ to be $2\beta_j$ times $g$'s contribution (see Step \ref{step:initpg}). The price of the unallocated goods is set to be zero. Then, bearing in mind property \textbf{P2}, we set agent-specific prices $q^j_g$s to ensure that for agent $j$ only goods with small-enough value are eligible for transfer (Steps \ref{step:qprice1} and \ref{step:qprice2}). Scaling $q^j_g$s appropriately for each agent $j$, the algorithm finds a demand set $D_j$ under $v_j$ (Step \ref{step:demand}). For each agent $j$, the candidate set $\widehat{D}_j$ is obtained by selecting a cardinality-wise minimal subset of $D_j$ of sufficiently high value (Steps \ref{step:indexk} and \ref{step:Dhat}). We will prove that, until the social welfare (under $\widehat{v}_j$s) of the maintained allocation $(Y_1, \ldots, Y_n)$ reaches a high-enough value, assigning $\widehat{D}_a$ to an agent $a$ (and removing the goods in $\widehat{D}_a$ from the other agents' bundles) increases the welfare (Step \ref{step:ifcondt}). That is, for any considered allocation $(Y_1, \ldots, Y_n)$ with social welfare  less than a desired threshold, the if-condition in Step \ref{step:ifcondt} necessarily holds; see Lemma \ref{lem:swmain}. This lemma will establish the main result of this section (Theorem \ref{thm:cappedsw} below) that lower bounds the social welfare---under $\widehat{v}_i$s---of the computed allocation. 

\begin{restatable}{theorem}{cappedsw} \label{thm:cappedsw}
Let $\instance{[n], \R', \{ v_i\}_{i \in [n]}}$ be a fair division instance in which there exists an allocation $(O_1, \ldots, O_n)$ and a subset of agents $\overline{A} \subseteq [n]$ that satisfy properties {\bf P1} and {\bf P2} mentioned above. Then, given $\XOS$ and demand oracle access to the ($\XOS$) valuations $v_i$s, Algorithm \ref{Alg:CapSW} computes (\textit{in polynomial time}) an allocation $(Y_1, \ldots, Y_n)$ such that
\begin{equation*}
    \sum_{j =1}^n \widehat{v}_j(Y_j) \geq \frac{2}{25} \sum_{ j \in \overline{A} } \widehat{v}_j(O_j).
\end{equation*}
\end{restatable}

Note that the welfare bound obtained in this theorem is with respect to the agents in $\overline{A}$. For our analysis, it suffices to have a guarantee of this form. It is, however, interesting to note that (under property \textbf{P1}) we also obtain a $13$-approximation for the optimal social welfare: $ \sum_{j =1}^n \widehat{v}_j(Y_j) \geq \frac{2}{25} \sum_{ j \in \overline{A} } \widehat{v}_j(O_j) \geq \frac{2}{25} \frac{26\sqrt{n}}{27} \geq \frac{\sqrt{n}}{13}$; recall that, by definition, each  function $\widehat{v}_j$ is upper bounded by $\frac{1}{\sqrt{n}}$ and, hence, the optimal social welfare under these functions is at most $\sqrt{n}$.

To prove Theorem \ref{thm:cappedsw}, we first establish the following lemmas.
\begin{lemma} \label{lem:val}
Throughout its execution, Algorithm \ref{Alg:CapSW} maintains 
\begin{equation*}
    \widehat{v}_j(Y_j) = \beta_j v_j(Y_j) < \frac{1}{\sqrt{n}} \qquad \text{ for all agents $j \in [n]$.}
\end{equation*}
\end{lemma}
\begin{proof}
Fix any agent $j \in [n]$ and consider any iteration in which $j$ receives set $\widehat{D}_j$ (i.e., Step \ref{step:iUpdate} executes with $a=j$). We will first show that $\widehat{v}_j (\widehat{D}_j) < \frac{1}{\sqrt{n}}$. Note that $\widehat{D}_j \subseteq D_j$, where $D_j$ is the demand set queried for agent $j$ in Step \ref{step:demand}; in particular, $D_j \in \argmax_{S} \left( v_j(S) - \sum_{g \in S} \frac{q^j_g}{\beta_j} \right)$. Hence, the goods in $D_j$ have finite prices, $q^j_g < \infty$. Consequently, for each $g \in D_j$, we have $\widehat{v}_j(g)\leq \frac{1}{2\sqrt{n}}$; see Steps \ref{step:qprice1} and \ref{step:qprice2}. Furthermore, given that $\widehat{D}_j$ is a minimal set (within $D_j$) with value at least $\frac{92}{225\sqrt{n}}$ (Steps \ref{step:indexk} and \ref{step:Dhat}) and $\widehat{v}_j$ is subadditive, we obtain $\widehat{v}_j(\widehat{D}_j)\leq \frac{92}{225\sqrt{n}}+\frac{1}{2\sqrt{n}}<\frac{1}{\sqrt{n}}$. 

This bound implies that throughout the subroutine's execution agent $j$ receives a bundle $Y_j$ of value (under $\widehat{v}_j$) less than $\frac{1}{\sqrt{n}}$: At the beginning of the subroutine $Y_j = \emptyset$, i.e., $\widehat{v}_j(Y_j) = 0$. Furthermore, between executions of Step \ref{step:iUpdate} specifically for  agent $j$, goods are only removed from $Y_j$. Therefore, monotonicity of the function $\widehat{v}_j$ ensures that $\widehat{v}_j(Y_j) < \sfrac{1}{\sqrt{n}}$ throughout the execution Algorithm \ref{Alg:CapSW}. 

By definition, $\widehat{v}_j(Y_j)=\min \left\{\frac{1}{\sqrt{n}}, \beta_j v_j(Y_j) \right\}$. Hence, the inequality $\widehat{v}_j(Y_j) < \frac{1}{\sqrt{n}}$ gives us $ \widehat{v}_j(Y_j) = \beta_j \ v_j(Y_j) < \frac{1}{\sqrt{n}}$. The lemma stands proved. 
\end{proof}

The next lemma bounds the loss in welfare due to reassignment of goods in Step \ref{step:reassign} of the subroutine. 

\begin{lemma}\label{lem:change}
In any while-loop iteration of Algorithm \ref{Alg:CapSW}, let $Y_j$ be the bundle assigned to any agent $j \in [n]$ and $p_g$s be the prices at the start of the iteration (i.e., in Step \ref{step:initpg}). Then, for all subsets $X \subseteq Y_j$ 
\begin{align*}
    \widehat{v}_j \left(Y_j \setminus X \right) & \ \geq \widehat{v}_j(Y_j) - \frac{1}{2} \sum_{g \in X} p_g.
\end{align*}
\end{lemma}
\begin{proof}
Note that $f_{j, Y_j}$ denotes the additive function that induces $v_j(Y_j)$ (see equation (\ref{defn:fiS})). Therefore, 
\begin{align}
\widehat{v}_j(Y_j) &= \beta_j v_j(Y_j) \tag{via Lemma \ref{lem:val}}\\
& = \beta_j \sum_{g\in Y_j}f_{j,Y_j}(g) \tag{by definition of $f_{j, Y_j}$} \\
&= \frac{1}{2} \sum_{g\in Y_j} p_g  \label{ineq:sumprice} 
\end{align}
The last inequality follows from how the prices, $p_g$s, were set in Step \ref{step:price} of Algorithm \ref{Alg:CapSW}. Furthermore, using the fact that $v_j$ is $\XOS$ we obtain   
\begin{align*}
\beta_j \  v_j(Y_j \setminus X) &\geq \beta_j \sum_{g\in Y_j \setminus X} f_{j,Y_j}(g) \\
&=\frac{1}{2}\sum_{g\in Y_j \setminus X} p_g \tag{by definition of $p_g$ in Step \ref{step:price}} \\
&= \frac{1}{2}\sum_{g\in Y_j } p_g - \frac{1}{2}\sum_{g\in  X} p_g \\
& = \widehat{v}_j(Y_j) - \frac{1}{2}\sum_{g\in  X} p_g  \tag{via (\ref{ineq:sumprice})}
\end{align*}

Since the function $\widehat{v}_j$ is monotone, for any subset $X \subseteq Y_j$, we have $\widehat{v}_j(Y_j\setminus X) \leq \widehat{v}_j(Y_j) <\frac{1}{\sqrt{n}}$; here, the last inequality follows from Lemma \ref{lem:val}. Therefore, $\widehat{v}_j(Y_j\setminus X)=\beta_j \ v_j(Y_j \setminus X)$. 
These observations establish the desired inequality:  
$\widehat{v}_j(Y_j\setminus X) \geq \widehat{v}_j(Y_j)-\frac{1}{2} \sum_{g\in X} p_g$.
\end{proof}

The next lemma shows that, in Algorithm \ref{Alg:CapSW}, the social welfare (under $\widehat{v}_j$s) of the maintained allocation $(Y_1, \ldots, Y_n)$ keeps on increasing till it reaches $\frac{2}{25}\sum_{j \in \overline{A}} \widehat{v}_j(O_j)$. That is, for any considered allocation $(Y_1, \ldots, Y_n)$ with social welfare  less than $\frac{2}{25} \sum_{j \in \overline{A}} \widehat{v}_j(O_j)$, the if-condition in Step \ref{step:ifcondt} necessarily holds.

\begin{lemma}\label{lem:swmain}
Let $(Y_1, \ldots, Y_n)$ be an allocation considered in any iteration of Algorithm \ref{Alg:CapSW} with the property that  
\begin{equation*}
    \sum_{j=1}^n \widehat{v}_j(Y_j) < \frac{2}{25} \sum_{j \in \overline{A} } \widehat{v}_j(O_j).
\end{equation*}
Then, there exists an agent $ a \in [n]$ such that
\begin{equation*}
    \widehat{v}_a(\widehat{D}_a) + \sum_{j \in [n] \setminus\{ a\}} \widehat{v}_j(Y_j \setminus \widehat{D}_a) \geq \sum_{j =1}^n \widehat{v}_j(Y_j) + \frac{1}{225 \sqrt{n}}.
\end{equation*}
Here, set $\widehat{D}_a$ is as defined in Step \ref{step:Dhat} of the algorithm. 
\end{lemma}

\begin{proof}
First, we express the social welfare of allocation $(Y_1, \ldots, Y_n)$ in terms of the prices $p_g$s (set in Step \ref{step:initpg})
\begin{align*}
\sum _{j =1}^n \widehat{v}_j(Y_j) & = \sum_{j=1}^n \beta_j \ v_j(Y_j) \tag{via Lemma \ref{lem:val}} \\
& = \sum_{j=1}^n \beta_j \sum_{g\in Y_j}  f_{j, Y_j} (g) \tag{by definition of $f_{j, Y_j}$} \\
& = \sum_{j =1}^n \sum_{g\in Y_j} \frac{p_g}{2} \tag{considering Step \ref{step:initpg}} \\
& = \sum_{g \in \R' - Y_0} \frac{p_g}{2}  \\
& = \sum_{g \in \R'} \frac{p_g}{2} \tag{$p_g = 0$, for all $g \in Y_0$; Step \ref{step:pricez}}
\end{align*}
Therefore, the lemma assumption, $\sum_{j=1}^n \widehat{v}_j(Y_j) < \frac{2}{25} \sum_{j \in \overline{A}} \ \widehat{v}_j (O_j)$, reduces to
\begin{align}
\frac{1}{2} \sum_{g \in \R'} p_g < \frac{2}{25} \sum_{j \in \overline{A}} \widehat{v}_j(O_j) \label{ineq:pricedom}
\end{align}

Multiplying both sides of inequality (\ref{ineq:pricedom}) by 2 gives us  
\begin{align}
\sum_{g \in \R'} p_g & < \frac{4}{25}\sum_{j \in \overline{A}} \widehat{v}_j(O_j) \nonumber \\
&= \sum_{j \in \overline{A}} \widehat{v}_j(O_j) -\frac{21}{25}\sum_{j \in \overline{A}} \widehat{v}_j(O_j)  \nonumber \\
&\leq \sum_{j \in \overline{A}} \widehat{v}_j(O_j) -\frac{21}{25}\left( \frac{26\sqrt{n}}{27}\right) \tag{from property \textbf{P1}} \\
& \leq \sum_{j \in \overline{A}} \widehat{v}_j(O_j) - \frac{182 \sqrt{n}}{225} \label{ineq:toomuch}
\end{align}

For any set $S$, write cumulative price $p(S) \coloneqq \sum_{g \in S} p_g$. Applying this notation and rearranging inequality (\ref{ineq:toomuch}) we get\footnote{Recall that the prices $p_g$s are nonnegative.}  
\begin{align}
\sum_{j \in \overline{A}}  \left( \widehat{v}_j(O_j)  - p(O_j) \right) \geq \frac{182 \sqrt{n}}{225} \label{ineq:drag}
\end{align}

Next, we define subset of agents 
\begin{align*}
H \coloneqq \left\{ h \in [n] :  \widehat{v}_h (Y_h) \geq \frac{1}{5\sqrt{n}} \right\}.
\end{align*}
The lemma assumption $\sum_{j=1}^n \widehat{v}_j(Y_j)< \frac{2}{25} \sum_{j \in \overline{A}} \  \widehat{v}_j (O_j)$ implies that $|H| < \frac{2n}{5}$. Otherwise, we would obtain a contradiction: $\sum_{h \in H} \widehat{v}_h(Y_h)\geq \frac{|H|}{5\sqrt{n}} \geq \frac{2\sqrt{n}}{25} \geq \frac{2}{25} \sum_{j \in \overline{A}} \  \widehat{v}_j(O_j)$. Recall that, by definition, the valuations $\widehat{v}_j$s are upper bounded by $\frac{1}{\sqrt{n}}$. 

Inequality (\ref{ineq:drag}) can be expressed as
\begin{align*}
\sum_{j\in \overline{A} \setminus H} (\widehat{v}_j(O_j) -p(O_j)) + \sum_{h \in \overline{A} \cap H} (\widehat{v}_h(O_h) -p(O_h)) \geq \frac{182\sqrt{n}}{225}. 
\end{align*}
Therefore, the inequality $\widehat{v}_h(O_h) \leq \frac{1}{\sqrt{n}}$ and the fact that prices are nonnegative, lead to 
\begin{align*}
    \sum_{j\in \overline{A} \setminus H} (\widehat{v}_j(O_j) -p(O_j)) &\geq \frac{182\sqrt{n}}{225}  - \frac{|\overline{A} \cap H|}{\sqrt{n}}\\ 
    &\geq \frac{182\sqrt{n}}{225}  - \frac{|H|}{\sqrt{n}}\\ 
    &\geq \frac{182\sqrt{n}}{225}  - \frac{2\sqrt{n}}{5} \tag{since $|H| < \frac{2n}{5}$} \\
    &= \frac{92\sqrt{n}}{225}.
\end{align*}

Hence, there exists an agent $a \in \overline{A} \setminus H$ such that 
\begin{align}
    \widehat{v}_a(O_a) -p(O_a) \geq \frac{1}{| \overline{A} \setminus H|} \ \frac{92\sqrt{n}}{225}\geq \frac{92}{225\sqrt{n}} \label{eq:existimprov}
\end{align}

We will complete the proof by showing that the lemma holds for this specific agent $a \in \overline{A} \setminus H$. Towards this, we first bound $\widehat{v}_a ( \widehat{D}_a)$ and show that this value is at least the price of the set $\widehat{D}_a$. 

Recall that set $D_a$ is obtained by the demand oracle for agent $a$ (i.e., for valuation $v_a$) under prices $\frac{q_g^a}{\beta_a}$ (Step \ref{step:demand}). Furthermore, for the agent $a \in \overline{A} \setminus H$ and all goods $g' \in O_a$, we have $\widehat{v}_a(g') \leq \frac{1}{2\sqrt{n}}$ (via property \textbf{P2}). Hence, all goods $g' \in O_a$ have finite prices that satisfy $q_{g'}^a = p_{g'}$ (Step \ref{step:qprice1}). Hence, $O_a$ a feasible set to be demanded and the demand optimality of $D_a$ gives us $v_a(D_a)- \sum_{g\in D_a} \left( \frac{p_g}{\beta_a} \right) \geq v_a(O_a) - \sum_{g\in O_a} \left( \frac{p_g}{\beta_a} \right)$. Multiplying throughout by $\beta_a >0$ we obtain  
\begin{align*}
    \beta_a v_a(D_a) -\sum_{g\in D_a} p_g & \geq \beta_a v_a(O_a) -\sum_{g\in O_a} p_g \\
    &\geq \widehat{v}_a(O_a) -\sum_{g\in O_a} p_g \tag{By definition of $\widehat{v}_a$}\\
    &\geq \frac{92}{225\sqrt{n}}  \tag{via (\ref{eq:existimprov})}
\end{align*}

Since the the prices are non-negative, $\beta_a v_a(D_a)\geq \frac{92}{225\sqrt{n}}$. That is, in Step \ref{step:Dhat} for agent $a$ the desired set $\widehat{D}_a \subseteq D_a$ can be found with value 
\begin{align}
\widehat{v}_a(\widehat{D}_a) \geq \frac{92}{225\sqrt{n}} \label{eq:Dhatbound}
\end{align}

In addition, the demand optimality of $D_a$ implies that all the goods $g \in D_a$ have finite prices. Therefore, $\widehat{v}_a(g) \leq \frac{1}{2\sqrt{n}}$ for all goods $g \in D_a \supseteq \widehat{D}_a$. This bound and the selection of $D_a$ gives us 

\begin{equation}
    \widehat{v}_a(\widehat{D}_a) \leq \frac{92}{225\sqrt{n}} + \frac{1}{2\sqrt{n}}<\frac{1}{\sqrt{n}} \label{eq:Dhatbound1}
\end{equation}

We will next show that $\widehat{v}_a(\widehat{D}_a)$ is at least the price of the set $\widehat{D}_a$. Write $f_{a, D_a}(\cdot)$ to denote the additive function that induces $v_a(D_a)$; in particular, $v_a(D_a) = \sum_{g \in D_a} f_{a, D_a} (g)$. The demand optimality of $D_a$, under the prices $\frac{q^a_g}{\beta_a}$, implies $f_{a,D_a}(g) - \frac{q_g^a}{\beta_a} \geq 0$ for all goods $g \in D_a$. Equivalently,  $\beta_a f_{a,D_a}(g) - q_g^a \geq 0$ for all goods $g \in D_a \supseteq \widehat{D}_a$.  Using the bound we obtain 
\begin{align*}
    \widehat{v}_a(\widehat{D}_a)-\sum_{g\in \widehat{D}_a} p_g &= \beta_a v_a(\widehat{D}_a) -\sum_{g\in \widehat{D}_a} p_g  \tag{via (\ref{eq:Dhatbound1}) and the definition of $\widehat{v}_a$} \\
    &\geq \sum_{g\in \widehat{D}_a} \beta_a f_{a,D_a} (g) - \sum_{g\in \widehat{D}_a} p_g \tag{since $v_a$ is $\XOS$} \\
   & =  \sum_{g\in \widehat{D}_a} \left(\beta_a f_{a,D_a} (g) -   p_g  \right) \geq 0.
\end{align*}
That is,
\begin{equation}
\widehat{v}_a(\widehat{D}_a) \geq \sum_{g\in \widehat{D}_a} p_g   \label{eq:hash}
\end{equation}
Using this inequality we can bound the change in social welfare when $\widehat{D}_a$ is assigned to agent $a$: 
\begin{align*}
    \widehat{v}_a(\widehat{D}_a) - \widehat{v}_a(Y_a) +\sum_{j \in [n] \setminus \{a\}} \left(\widehat{v}_j(Y_j \setminus \widehat{D}_a) -\widehat{v}_j(Y_j)\right)&\geq \left( \widehat{v}_a(\widehat{D}_a)-\widehat{v}_a(Y_a)\right) - \frac{1}{2} \sum_{g\in \widehat{D}_a} p_g   \tag{via Lemma \ref{lem:change}}\\
    &\geq  \left( \widehat{v}_a(\widehat{D}_a)-\widehat{v}_a(Y_a)\right) - \frac{\widehat{v}_a (\widehat{D}_a)}{2} \tag{via (\ref{eq:hash})}\\
    &= \frac{1}{2}\widehat{v}_a(\widehat{D}_a)-\widehat{v}_a(Y_a) \\
& \geq \frac{92}{450 \sqrt{n}} -\widehat{v}_a(Y_a) \tag{via (\ref{eq:Dhatbound})} \\
& \geq \frac{92}{450 \sqrt{n}}  - \frac{1}{5\sqrt{n}} \tag{since $a \in \overline{A} \setminus H$, i.e., $a \notin H$} \\
&=\frac{1}{225 \sqrt{n}}.
\end{align*}
Hence, for agent $a$ we necessarily obtain the desired increase in social welfare:
\begin{equation*}
    \widehat{v}_a(\widehat{D}_a) + \sum_{j \in [n] \setminus\{ a\}} \widehat{v}_j(Y_j \setminus \widehat{D}_a) \geq \sum_{j =1}^n \widehat{v}_j(Y_j) + \frac{1}{225 \sqrt{n}}.
\end{equation*}
This completes the proof. 
\end{proof}

We now restate and prove Theorem \ref{thm:cappedsw}.

\cappedsw*

\begin{proof}
The contrapositive of Lemma \ref{lem:swmain}, implies that if there does not exist an agent $a$ such that $\widehat{v}_a(\widehat{D}_a) +\sum_{j\neq a} \widehat{v}_j(Y_j -\widehat{D}_a) \geq \sum_{j=1}^n \widehat{v}_j(Y_j) +\frac{1}{225 \sqrt{n}}$ (i.e., the if-condition in Step \ref{step:ifcondt} is \emph{not} satisfied), then $\sum_{j =1}^n \widehat{v}_j(Y_j) \geq \frac{2}{25} \sum_{ j \in \overline{A} } \widehat{v}_j(O_j)$. Therefore, the algorithm  terminates only when we have the desired approximation to the welfare among agents in $\overline{A}$. This establishes the correctness of Algorithm \ref{Alg:CapSW}. 

For the run-time analysis, note that in every iteration of the while-loop in the algorithm the social welfare increases by at least $\frac{1}{225\sqrt{n}}$. Since the functions $\widehat{v}_j$s are upper bounded by $\frac{1}{\sqrt{n}}$, the maximum possible social welfare is $\sqrt{n}$. Hence, the while-loop iterates at most $225n$. Given that each iteration of the loop executes in polynomial time (using value, demand, and $\XOS$ oracle access to the valuations $v_j$s), we get that the algorithm computes an allocation in polynomial time. This establishes the theorem.  
\end{proof}

\section{Sublinear Approximation Algorithm for Nash Social Welfare}\label{sec:analysis}
The section establishes our main result, the approximation ratio of Algorithm \ref{Alg:NSWXOS} for Nash social welfare, through a baroque case analysis. Recall that Algorithm \ref{Alg:NSWXOS} first removes $n\log n$ goods by taking repeated matchings. Then, the remaining goods are randomly partitioned into subsets $\R$ and $\R'$. A discrete moving knife subroutine is executed over the goods in $\R$ (Algorithm \ref{Alg:MovingKnife}) and Algorithm \ref{Alg:CapSW} partitions the goods in $\R'$ to (approximately) maximize social welfare under the capped valuations $\widehat{v}_i$s.

\mainthm*

Recall that $\Q = (\mu(i) + \pi(i) + X_i + Y_i)_{i \in [n]}$ denotes the allocation returned by Algorithm \ref{Alg:NSWXOS}; here we use notation as in Algorithm \ref{Alg:NSWXOS}. Also, as before, $\mathcal{N} = (N_1, \ldots, N_n)$ denotes a Nash optimal allocation and $g^*_i = \argmax_{g \in N_i} v_i(g)$ for all agents $i \in [n]$. For analytic purposes, we will consider the allocation in which the goods $g^*_i$ are included in the bundles $Q_i = \mu(i) + \pi(i) + X_i + Y_i$, in lieu of the goods $\mu(i)$; specifically, throughout this section write $Q^*_i \coloneqq g_i^* + \pi(i) + X_i + Y_i$, for all agents $i \in [n]$, and allocation $\mathcal{Q}^* \coloneqq (Q^*_1, \ldots, Q^*_n)$.

The next lemma (proved in Appendix \ref{appendix:h-to-qstar}) shows that the Nash social welfare of allocation $\Q$ is within a factor of $1/2$ of the Nash social welfare of the allocation $\Q^*$. 
\begin{restatable}{lemma}{LemmaReduce}
\label{lem:reduce}
Let $\Q = \left(\mu(i) + \pi(i) + X_i + Y_i\right)_{i \in [n]}$ denote the allocation computed by Algorithm \ref{Alg:NSWXOS} and write allocation $\Q^*=(Q^*_1, \ldots, Q^*_n)$ with bundles $Q^*_i \coloneqq g_i^* + \pi(i) + X_i + Y_i$, for all $i \in [n]$. Then, 
\[    \NSW (\Q) \geq \frac{1}{2}\NSW(\Q^*). \]\\
\end{restatable}

\begin{figure}[h!]
    \centering
    \includegraphics[width=\linewidth]{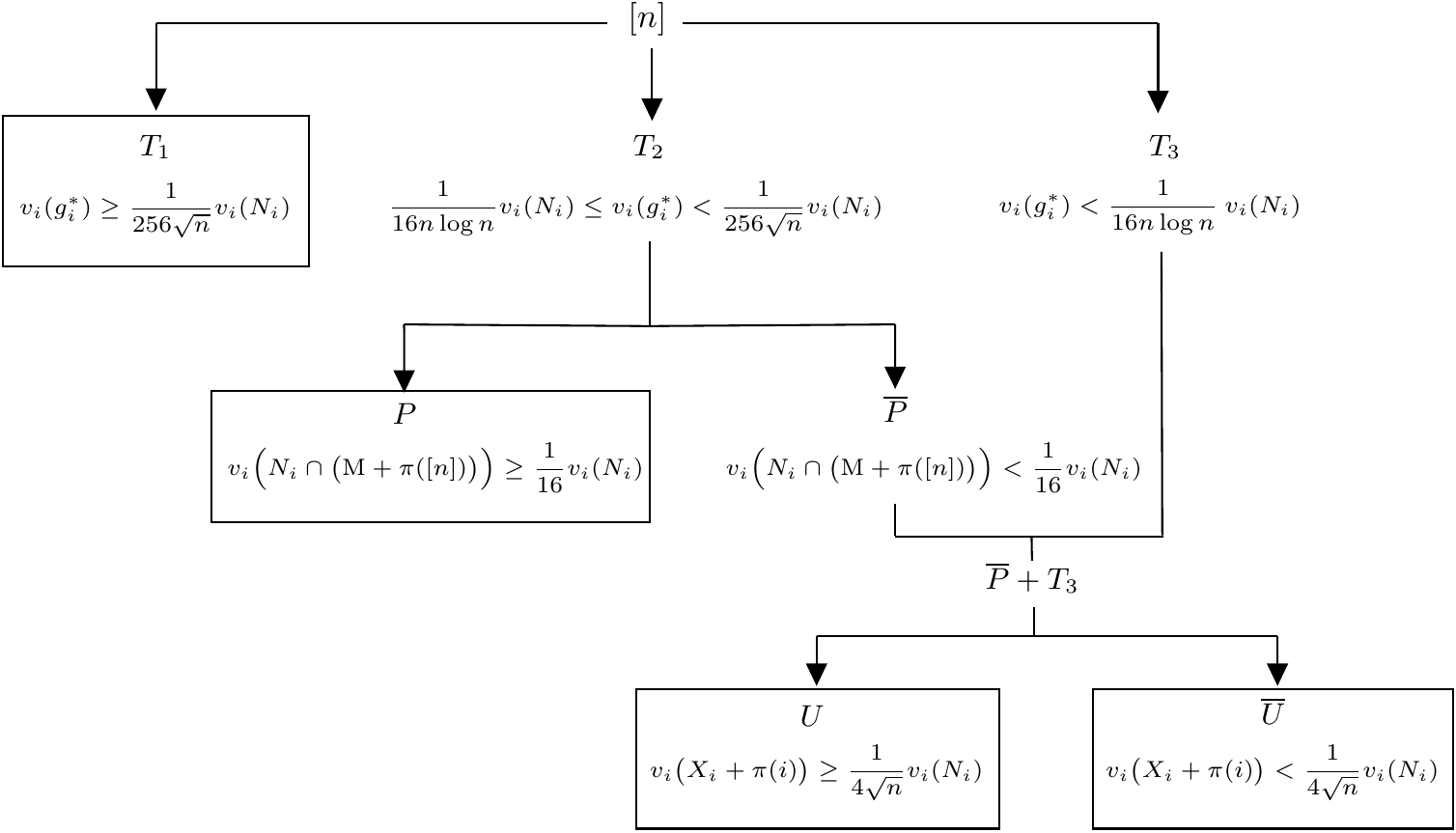}
    \caption{\footnotesize The figure shows the partitions (of the set of agents $[n]$) used in the analysis of Algorithm \ref{Alg:NSWXOS}. Recall that $(N_1, \ldots, N_n)$ is a Nash optimal allocation for the give instance. $(X_1, \ldots, X_n)$ is the allocation returned by Algorithm $\ref{Alg:MovingKnife}$ (\textsc{DiscreteMovingKinfe}), $(Y_1, \ldots, Y_n)$ is the allocation returned by Algorithm \ref{Alg:CapSW} (\textsc{CappedSocialWelfare}) and $g_i^* \coloneqq \argmax_{g \in N_i} v_i(g)$. Also recall that $\M$ and $\{\pi(i)\}_{i \in [n]}$ are the sets of goods matched in Phase {\rm I} of Algorithm \ref{Alg:NSWXOS}.}
    \label{fig:cases}
\end{figure}

For a case analysis, we fist partition the agents into different types, $T_1$, $T_2$, and $T_3$, depending on the value they have for their $g_i^*$; see Figure \ref{fig:cases}. Specifically, 
\begin{align*}
T_1 & \coloneqq \left\{ i \in [n] \ : \ v_i(g^*_i) \geq \frac{1}{256 \sqrt{n}} v_i(N_i) \right\}, \\ 
T_2 & \coloneqq \left\{ i \in [n] \ : \frac{1}{16 n \log n} v_i(N_i) \leq v_i(g^*_i) < \frac{1}{256 \sqrt{n}} v_i(N_i) \right\}, \quad \text{ and } \\ 
T_3 & \coloneqq \left\{  i \in [n] \ : v_i(g^*_i)  < \frac{1}{16 n \log n} v_i(N_i)  \right\}. 
\end{align*}
Note that agents in $T_1$ achieve an $O(\sqrt{n})$ approximation if they receive their optimal goods $g_i^*$s. To show that most other agents also get a sublinear approximation, we sub-divide the sets $T_2$ and $T_3$ based on the subsets computed in Algorithm \ref{Alg:NSWXOS}. 

In particular, we partition $T_2$ into two subsets: $P \coloneqq \left\{ i \in T_2 \ : \ v_i \big( N_i \cap ( \M + \pi( [n] ) ) \big) \geq \frac{1}{16} v_i(N_i) \right\}$ and $\overline{P} \coloneqq \left\{ i \in T_2 \ : \ v_i \big( N_i \cap ( \M + \pi( [n] ) ) \big) < \frac{1}{16} v_i(N_i) \right\}$. Note that the valuation $v_i$ is $\XOS$ (subadditive), hence, for all agents $i \in \overline{P}$ we have $v_i \big( N_i \setminus ( \M + \pi( [n] ) ) \big) \geq \frac{15}{16} v_i(N_i)$. 

It will also be helpful to consider the following partition of $\overline{P} + T_3$, based on the values obtained by $X_i$ (the bundle computed by the discrete moving knife procedure) and the matched good $\pi(i)$. 
 \begin{align*}
U & \coloneqq \left\{ i \in \overline{P} + T_3 \ : \ v_i(X_i + \pi(i)) \geq \frac{1}{4 \sqrt{n}} v_i(N_i) \right\}. \\
\overline{U} & \coloneqq \left\{ i \in \overline{P} + T_3 \ : \ v_i(X_i + \pi(i)) < \frac{1}{4 \sqrt{n}} v_i(N_i) \right\}.
\end{align*}

The remainder of the section considers the following two exhaustive cases and shows that in both we achieve the desired approximation ratio of Algorithm \ref{Alg:NSWXOS}. 
\begin{enumerate}
    \item[] \textbf{Case \textcal{1}}: $|T_1 + P + U| \geq \frac{n}{27}$
    \item[] \textbf{Case \textcal{2}}: $|\overline{U}| \geq \frac{26n}{27}$
\end{enumerate}

Specifically, in both cases, we show that the allocation $\Q^*$ (and consequently the computed allocation $\Q$) achieves a sublinear approximation to $\NSW(\N)$. The rest of the proof is structured as follows:
\begin{itemize}
    \item Subsection \ref{subsection:tipu}: \ Lemmas \ref{lem:sublinearT1}, \ref{lem:sublinearP}, and \ref{lem:sublinearU} prove that---in allocation $\Q^*$---the agents in the sets $T_1$, $P$, and $U$, respectively, achieve a sublinear approximation.
    \item Subsection \ref{subsection:ToP}: \ Lemmas \ref{lem:linearT3} and \ref{lem:linearPbar} show that agents in $T_3$ and $\overline{P}$, respectively, achieve a linear approximation, even when we restrict attention to goods assigned in the first three phases of Algorithm \ref{Alg:NSWXOS}. 
    \item Subsection \ref{subsection:caseone}: \ Lemma \ref{lem:case1} shows that in Case \textcal{1} (i.e., when $|T_1 + P + U| \geq \sfrac{n}{27}$), Algorithm \ref{Alg:NSWXOS} obtains an ${O}(n^{\nicefrac{53}{54}})$-approximation ratio for Nash social welfare maximization.  
    \item Subsection \ref{subsection:propalgthree}: \ Moving on to Case \textcal{2} ($|\overline{U}| \geq \frac{26n}{27}$), Propositions \ref{lem:highsw} and \ref{lem:smallgoods} build upon Lemmas \ref{lem:linearT3} and \ref{lem:linearPbar}. They show that, for Case \textcal{2}, the properties required to apply \textsc{CappedSocialWelfare} subroutine hold.
    \item Subsection \ref{subsection:casetwo}: \ Lemma \ref{lem:case2} proves that (under allocation $\Q^*$) we get a sublinear approximation in Case \textcal{2} as well. Intuitively, the lemma shows that here the linear approximation (achieved for agents in $\overline{U} \subseteq T_3 + \overline{P}$ in the first three phases) gets bootstrapped to a sublinear approximation in the fourth phase of Algorithm \ref{Alg:NSWXOS} (i.e., in the \textsc{CappedSocialWelfare} subroutine). 
    
    \item Subsection \ref{subsection:mainproof}: \ Theorem \ref{thm:main} finally follows from Lemmas \ref{lem:reduce}, \ref{lem:case1} and \ref{lem:case2}.
\end{itemize}

\subsection{Sublinear Approximation for agents in $T_1$, $P$, and $U$}
\label{subsection:tipu}

\begin{lemma}\label{lem:sublinearT1}
    For each agent $i \in T_1$ we have $v_i(Q_i^*) \geq \frac{1}{256\sqrt{n}}v_i(N_i)$.
\end{lemma}
\begin{proof}
For each agent $i \in T_1$, by definition, $v_i(g^*_i) \geq \frac{1}{256\sqrt{n}} v_i(N_i)$. Hence, via the monotonicity of valuation $v_i$, the desired bound follows: $v_i(Q_i^*) = v_i(g_i^* + \pi(i) + X_i + Y_i) \geq v_i(g_i^*) \geq \frac{1}{256\sqrt{n}} v_i(N_i)$.
\end{proof}
The next lemma addresses agents in the set $P = \left\{ i \in T_2 \ : \ v_i \big( N_i \cap ( M + \pi( [n] ) ) \big) \geq \frac{1}{16} v_i(N_i) \right\}$. 
\begin{lemma}\label{lem:sublinearP}
For agents in the set $P$ we have  
\begin{align*} 
    \left(\prod_{i \in P} v_i(Q_i^*)\right)^{\nicefrac{1}{n}} \geq \frac{1}{2} \left( \prod_{i \in P} \frac{1}{\sqrt{n}} v_i(N_i) \right)^{\nicefrac{1}{n}}.
\end{align*}
\end{lemma}
\begin{proof}
For each agent $i \in P$, write $k_i \coloneqq |N_i \cap (\M+\pi([n]))|$ where $\M$ is the set of goods identified in Phase {\rm I} of Algorithm $\ref{Alg:NSWXOS}$ and $\pi(\cdot)$ is the matching computed in Step \ref{step:pi}. Note that  $\sum_{i \in P} k_i \leq n(\log{n}+1) $, since  $|\M|\leq n\log{n}$ and $|\pi ([n])|\leq n$. Here, the AM-GM inequality gives us 
\begin{align*}
\left(\prod_{i \in P} k_i \right)^{\sfrac{1}{n}} \leq \left(\frac{1}{|P|}\sum_{i \in P} k_i\right)^{\sfrac{|P|}{n}} \leq \left( \frac{2n\log{n}}{|P|} \right)^{\sfrac{|P|}{n}} = \left( \frac{n}{|P|}\right)^{\sfrac{|P|}{n}} \left( 2 \log n \right)^{\sfrac{|P|}{n}} \leq e^{\nicefrac{1}{e}} \left( 2 \log n \right)^{\sfrac{|P|}{n}}.
\end{align*}
The last inequality follows from the fact that $x^{\sfrac{1}{x}}$ is maximized at $x = e$; here, $x = \nicefrac{n}{|P|} \geq 1$. Therefore, 
\begin{align}
\left(\prod\limits_{i\in P }\frac{1}{k_i}\right)^{\sfrac{1}{n}} \geq \frac{1}{e^{\nicefrac{1}{e}}} \left(\frac{1}{2 \log{n}}\right)^{\sfrac{|P|}{n}} \geq \frac{1}{2} \left(\frac{1}{2 \log{n}}\right)^{\sfrac{|P|}{n}} \label{ineq:denomgm}
\end{align}

Furthermore, note that, for any agent $i \in P$, using the subadditivity of $v_i$ and the definition of set $P$, we get 
	\begin{align}
		v_i(g^*_i) \geq \frac{1}{k_i}v_i \left(N_i\cap \big(\M+\pi([n]) \big) \right)  \geq \frac{1}{16k_i}v_i(N_i) \label{ineq:kini}
	\end{align}
These observations lead to the desired bound
\begin{align*}
    \left(\prod_{i \in P} v_i(Q_i^*)\right)^{\sfrac{1}{n}} &\geq \left(\prod_{i \in P} v_i(g_i^*)\right)^{\sfrac{1}{n}} \tag{$g^*_i \in Q^*_i$ and $v_i$ is monotonic} \\
    &\geq \left(\prod_{i \in P}\frac{1}{16k_i}v_i(N_i)\right)^{\sfrac{1}{n}} \tag{via (\ref{ineq:kini})} \\
    &= \left(\prod_{i \in P}\frac{1}{k_i}\right)^{\sfrac{1}{n}} \left(\prod_{i \in P} \frac{1}{16} v_i(N_i)\right)^{\sfrac{1}{n}} \\
& \geq \frac{1}{2} \left(\frac{1}{2 \log{n}}\right)^{\sfrac{|P|}{n}} \left(\prod_{i \in P} \frac{1}{16} v_i(N_i)\right)^{\sfrac{1}{n}}   \tag{via (\ref{ineq:denomgm})} \\
& = \frac{1}{2} \left(\prod_{i \in P} \frac{1}{32 \log n} v_i(N_i)\right)^{\sfrac{1}{n}}  \\
&\geq   \frac{1}{2} \left( \prod_{i \in P} \frac{1}{\sqrt{n}} v_i(N_i) \right)^{\nicefrac{1}{n}}.
\end{align*}
The last inequality holds for a moderately large $n$. The lemma stands proved. 
\end{proof}
The following lemma addresses agents in the set $U = \{ i \in \overline{P} + T_3 \ : \ v_i(X_i + \pi(i)) \geq \frac{1}{4 \sqrt{n}} v_i(N_i) \}$.
\begin{lemma}\label{lem:sublinearU}
For each agent $i \in U$ we have $v_i(Q_i^*) \geq \frac{1}{4\sqrt{n}}v_i(N_i)$.
\end{lemma}
\begin{proof}
The lemma follows directly from the definition of $U$:
\begin{align*}
    v_i(Q_i^*) &= v_i(g_i^* + X_i + \pi(i) + Y_i) \tag{since $Q_i^* = g_i^* + X_i + \pi(i) + Y_i$} \\
    &\geq v_i(X_i + \pi(i)) \tag{$v_i$ is monotonic} \\
    &\geq \frac{1}{4\sqrt{n}} v_i(N_i) \tag{by definition of $U$}
\end{align*}
\end{proof}

\subsection{Linear Approximation for $T_3$ and $\overline{P}$}
\label{subsection:ToP}
This section shows that even if we restrict attention to the assignments made before the fourth phase in Algorithm \ref{Alg:NSWXOS} (i.e., if we consider $(X_i + \pi(i))$s), with high probability, each agent $i \in T_3 + \overline{P}$ achieves a linear approximation with respect to $v_i(N_i)$. 

\begin{lemma}\label{lem:linearT3}
For each agent $i\in T_3$, we have (with high probability) $v_i(X_i)\geq \frac{1}{\constlinear n} v_i(N_i)$, where $(X_1,  \ldots, X_n)$ is the allocation returned by the \textsc{DiscreteMovingKnife} subroutine in Step \ref{step:MK} of Algorithm \ref{Alg:NSWXOS}.
\end{lemma}
\begin{proof} 
Fix any agent $i \in T_3$. We will first show that the previously-mentioned concentration bound (Lemma \ref{lem:concbound}) holds for agent $i \in T_3$, with  the random subset $\R$ drawn by Algorithm \ref{Alg:NSWXOS}. Using this, we will further establish that (with high probability) the condition required to successfully execute the \textsc{DiscreteMovingKnife} subroutine holds (i.e., the condition required to invoke Lemma \ref{lem:MK} holds) for agent $i \in T_3$. Consequently, applying Lemma \ref{lem:MK}, we will obtain the desired inequality $v_i(X_i)\geq \frac{1}{\constlinear n} v_i(N_i)$. 

We start by proving that $v_i\left( N_i \setminus \big(\M+\pi([n]) \big) \right) \geq \frac{7}{8}v_i(N_i)$. The definition of the set $T_3$ give us $v_i(g^*_i)\leq \frac{1}{\constTthree n\log{n}}v_i(N_i)$. 
	Hence,  
	\begin{align*}
		v_i\left( N_i\cap \big(\M+\pi([n]) \big) \right)&\leq \left| N_i\cap \big(\M+\pi([n]) \big) \right| v_i(g^{*}_i) \tag{$v_i$ is subadditive}\\
		&\leq \left( n\log{n} + n \right) \frac{1}{16 n\log{n}} v_i(N_i) \tag{$|M|\leq n\log{n}$ and $|\pi([n])| \leq n$}\\
		& \leq   2n\log{n}  \ \frac{1}{16 n\log{n}} v_i(N_i) \\
		&= \frac{1}{8}v_i(N_i).
	\end{align*}
Using this inequality and the subadditivity of $v_i$, we get 
\begin{align}
v_i\left( N_i \setminus \big(\M+\pi([n]) \big) \right) \geq \frac{7}{8}v_i(N_i) \label{ineq:hmm}
\end{align} 

Recall that subset $\R$ is obtained by randomly partitioning $[m]\setminus (\M+\pi([n]))$. Now, we can apply Lemma \ref{lem:concbound} with the set of goods $\G = [m]\setminus (\M+\pi([n]))$, the bundles $\oveN_j =  N_j \setminus (\M+\pi([n]))$  (for all agents $j \in [n]$), and $T=T_3$. In particular, Lemma \ref{lem:concbound} gives us\footnote{Since $\R \subseteq [m] \setminus (\M+\pi([n]))$ and $\oveN_i =  N_i \setminus (\M+\pi([n]))$, the following equalities hold $N_i \cap \R = \left( N_i \setminus \big(\M+\pi([n]) \big) \right) \cap \R = \oveN_i \cap \R$.}
\begin{align}
v_i(N_i\cap \R) \geq \frac{1}{3} v_i\left( \oveN_i \right) \underset{\text{via } (\ref{ineq:hmm})}{\geq}\frac{7}{24} v_i(N_i) \label{eqn:T3conc}
\end{align}
Furthermore, we can show that all the goods in $N_i \cap \R$ are of sufficiently small value; specifically, all goods $g \in N_i \cap \R$ satisfy 
\begin{align}
v_i(g) \underset{\text{since $i \in T_3$}}{\leq} \frac{1}{\constTthree n\log{n}} v_i(N_i) \underset{\text{via (\ref{eqn:T3conc})}}{\leq} \frac{1}{16  n\log{n}}\frac{24}{7} \ v_i(N_i\cap \R) \leq \frac{1}{16n} v_i(N_i \cap \R) \label{ineq:toomany}
\end{align}
The last inequality holds assuming a moderately high value of $n$. 

With the above-mentioned bounds, we can instantiate Lemma \ref{lem:MK} over instance $\langle [n], \R, \{v_i\}_{i \in [n]}\rangle$, bundles $\widetilde{N}_j = N_j \cap \R$ (for all $j \in [n]$) and $T = T_3$; note that inequality (\ref{ineq:toomany}) ensures that Lemma \ref{lem:MK}'s requirement for set $T$ is satisfied. Here, Lemma \ref{lem:MK} implies that the allocation $(X_1, \ldots, X_n)$---computed by the \textsc{DiscreteMovingKnife} subroutine in Algorithm \ref{Alg:NSWXOS}---satisfies $v_i(X_i) \geq \frac{1}{\constMK n} v_i(N_i\cap \R)$ for $i \in T_3$. 

Combining this inequality with equation (\ref{eqn:T3conc}), we obtain the desired guarantee 
\begin{align*}
v_i(X_i)\geq \frac{1}{\constMK n} v_i(N_i\cap \R) \geq \frac{1}{\constlinear n}v_i(N_i).
\end{align*}
This completes the proof. 
\end{proof}

Recall the definitions of sets $P$ and $\overline{P}$, and that they partition $T_2$:
\begin{align*}
P & = \left\{ i \in T_2 \ : \ v_i \big( N_i \cap ( M + \pi( [n] ) ) \big) \geq \frac{1}{16} v_i(N_i) \right\} \qquad \text{ and} \\
\overline{P} & = \left\{ i \in T_2 \ : \ v_i \big( N_i \cap ( M + \pi( [n] ) ) \big) < \frac{1}{16} v_i(N_i) \right\}.
\end{align*}

As mentioned previously, for all agents $i \in \overline{P}$, we have $v_i \big( N_i \setminus ( M + \pi( [n] ) ) \big) \geq \frac{15}{16} v_i(N_i)$. We prove the following lemma for the set $\overline{P}$. 

\begin{lemma}\label{lem:linearPbar}
For every agent $i\in \overline{P}$, we have (with high probability) $v_i(X_i +\pi(i))\geq \frac{1}{\constlinear n}
v_i(N_i)$, where $(X_1, \ldots, X_n)$ is the allocation returned by the \textsc{DiscreteMovingKnife} subroutine in Step \ref{step:MK} of Algorithm \ref{Alg:NSWXOS} and $\pi(\cdot)$ is the matching computed in Step \ref{step:pi} of the algorithm. 
\end{lemma}
\begin{proof} 
Fix any agent $i \in \overline{P}$. We will first show that the previously-mentioned concentration bound (Lemma \ref{lem:concbound}) holds for agent $i \in \overline{P}$, with  the random subset $\R$ drawn by Algorithm \ref{Alg:NSWXOS}. This essentially would imply that the $v_i(N_i \cap \R)$ is within a constant factor of $v_i(N_i)$. Then, we will consider two exhaustive cases: either the matched good $\pi(i)$ itself provides a linear approximation to $v_i(N_i)$, or it does not. In the first case, the desired linear approximation is directly achieved. For the second case we will show that (with high probability) the condition required to successfully execute the \textsc{DiscreteMovingKnife} subroutine holds (i.e., the condition required to invoke Lemma \ref{lem:MK} holds) for agent $i \in \overline{P}$. Consequently, applying Lemma \ref{lem:MK}, we will obtain the desired inequality in this case as well. 

Towards showing that Lemma \ref{lem:concbound} can be applied for agent $i \in \overline{P}$, note that all the goods $g \in  N_i \setminus \big( \M+ \pi([n]) \big)$ satisfy 
\begin{align}
v_i(g) & <  \frac{1}{\constTone\sqrt{n}}\ v_i(N_i)  \tag{since $i \in \overline{P} \subseteq T_2$} \\
& \leq \frac{1}{\constTone\sqrt{n}} \cdot \frac{16}{15}  \  v_i \left( N_i \setminus \big( \M+ \pi([n]) \big) \right) \tag{since $i \in \overline{P}$} \\
&\leq \frac{1}{240\ \sqrt{n}} \  v_i \left( N_i \setminus \big( \M+ \pi([n]) \big) \right) \label{ineq:nef}
\end{align}
Hence, we can apply Lemma \ref{lem:concbound} with the set of goods $\G = [m]\setminus (\M+\pi([n]))$, bundles $\oveN_j = N_j \setminus (\M+\pi([n]))$ (for all agents $j \in [n]$), and $T= \overline{P}$; in particular, inequality (\ref{ineq:nef}) ensures that the  Lemma \ref{lem:concbound}'s requirement for set $T= \overline{P}$ is satisfied. Lemma \ref{lem:concbound} gives us 
\begin{align}\label{eq:concPbar}
v_i(N_i\cap \R) \geq \frac{1}{3}v_i(N_i\setminus(\M+\pi([n]))) \geq \frac{1}{4}v_i(N_i)
\end{align}
The last inequality follows from the fact that $i \in \overline{P}$.

We now conduct a case analysis based on the value of the matched good $\pi(i)$. 

\noindent
{\it Case (i):} $v_i(\pi(i))\geq \frac{1}{16 n} v_i(N_i \cap \R)$. Here, the desired bound directly holds: 
\begin{align*}
v_i(X_i+\pi(i)) \geq v_i(\pi(i)) \geq \frac{1}{16n} v_i(N_i\cap \R) \underset{\text{via } (\ref{eq:concPbar})}{\geq} \frac{1}{\constlinear n}v_i(N_i).
\end{align*}
	
\noindent
{\it Case (ii):} $v_i(\pi(i))<\frac{1}{16 n}v_i(N_i\cap \R)$. In this case, for all $\tilde{g} \in [m] \setminus \left( \M+\pi([n]) \right) \supseteq (N_i\cap \R)$, we have 
\begin{align*}
v_i(\tilde{g}) \leq v_i(\pi(i)) < \frac{1}{16n}v_i(N_i\cap \R).
\end{align*}
Here, the first inequality follows from the fact that $\pi(\cdot)$ is maximum-product matching from $n$ to $[m]\setminus M$ (Step \ref{step:mu} in Algorithm \ref{Alg:NSWXOS}). Otherwise, if there exists a good $\tilde{g}\in [m]\setminus (\M+\pi([n]))$ such that $v_i(\tilde{g})>v_i(\pi(i))$, then replacing $\pi(i)$ with $\tilde{g}$ would improve upon $\pi$, contradicting its optimality. 

Hence, we can instantiate Lemma \ref{lem:MK} with instance $\langle [n], \R, \{v_i\}_{i \in [n]}\rangle$, bundles $\widetilde{N}_j  = N_i\cap \R$, and set of agents $T = \{i \in \overline{P} \ : \ v_i(\pi(i)) < \frac{1}{16 n}v_i(N_i\cap \R)\}$. Lemma \ref{lem:MK} ensures that the computed allocation $(X_1, \ldots, X_n)$ satisfies $v_i(X_i)\geq \frac{1}{\constMK n} v_i(N_i\cap \R)$ for agent $i$. Therefore, the stated inequality holds in the current context as well:
	\begin{align*}
		v_i(X_i+\pi(i)) \geq v_i(X_i) \geq  \frac{1}{\constMK n} v_i(N_i\cap \R) \underset{\text{via } (\ref{eq:concPbar})}{\geq}\frac{1}{\constlinear n}v_i(N_i).
	\end{align*}
\end{proof}

\subsection{Sublinear Approximation Guarantee in Case \textcal{1}}
\label{subsection:caseone}

This section establishes a subliear approximation ratio\footnote{Recall that for the computed allocation $\Q$ we have $\NSW (\Q) \geq \frac{1}{2}\NSW(\Q^*)$ (Lemma \ref{lem:reduce}).} for Algorithm \ref{Alg:NSWXOS} under Case \textcal{1}: $|T_1 + P + U| \geq \frac{n}{27}$. 
\begin{lemma}\label{lem:case1}
If $|T_1 + P + U| \geq \frac{n}{27}$, then the Nash social welfare of allocation $\Q^*$ is at least $\frac{c'}{n^{\nicefrac{53}{54}}}$ times the optimal Nash social welfare,
\begin{align*}
\NSW(\Q^*) \geq \frac{c'}{n^{\nicefrac{53}{54}}}  \NSW(\mathcal{N}).
\end{align*}
Here, $c' \in \mathbb{R}_+$ is a fixed constant. 
\end{lemma}
\begin{proof}
The sets $T_1$, $P$, $U$, and $\overline{U}$ constitute a partition of the set of agents $[n]$ (see Figure \ref{fig:cases}). Therefore,
\begin{align*}
  \NSW(\Q^*) & =  \left(\prod_{i \in [n]} v_i(Q_i^*)\right)^{\frac{1}{n}} \\ 
  &= \left(\prod_{i \in T_1} v_i(Q_i^*)\right)^{\frac{1}{n}} \left(\prod_{i \in P} v_i(Q_i^*)\right)^{\frac{1}{n}} \left(\prod_{i \in U} v_i(Q_i^*)\right)^{\frac{1}{n}} \left(\prod_{i \in \overline{U}} v_i(Q_i^*)\right)^{\frac{1}{n}} \\
  &\geq \left(\prod_{i \in T_1} \frac{1}{256\sqrt{n}} v_i(N_i)\right)^{\frac{1}{n}} \frac{1}{2}  \left(\prod_{i \in P} \frac{1}{\sqrt{n}} v_i(N_i)\right)^{\frac{1}{n}} \left(\prod_{i \in U} \frac{1}{4\sqrt{n}}v_i(N_i)\right)^{\frac{1}{n}}  \left(\prod_{i \in \overline{U}} v_i(Q_i^*)\right)^{\frac{1}{n}} \tag{via Lemmas \ref{lem:sublinearT1}, \ref{lem:sublinearP}, and \ref{lem:sublinearU}} \\
    & \geq \frac{1}{2}  \left(\prod_{i \in T_1} \frac{1}{256\sqrt{n}} v_i(N_i)\right)^{\frac{1}{n}} \left(\prod_{i \in P} \frac{1}{\sqrt{n}} v_i(N_i)\right)^{\frac{1}{n}} \left(\prod_{i \in U} \frac{1}{4\sqrt{n}}v_i(N_i)\right)^{\frac{1}{n}} \left(\prod_{i \in \overline{U}} \frac{1}{64n}v_i(N_i)\right)^{\frac{1}{n}} \tag{via Lemmas \ref{lem:linearT3} \& \ref{lem:linearPbar}; $\overline{U} \subseteq T_3 + \overline{P}$}\\
    &\geq \frac{1}{2} \left(\frac{1}{256\sqrt{n}}\right)^{\frac{|T_1+P+U|}{n}} \left(\frac{1}{64n}\right)^{\frac{|\overline{U}|}{n}} \left(\prod_{i \in [n]}v_i(N_i)\right)^{\sfrac{1}{n}} \\
    &\geq \frac{1}{2} \ \frac{1}{256} \left(\frac{1}{n} \right)^{\frac{1}{2} \cdot \frac{|T_1+P+U|}{n}}  \left( \frac{1}{n} \right)^{1 - \frac{|T_1+P+U|}{n}} \ \NSW({\mathcal{N}}) \tag{since $\frac{|\overline{U}|}{n} = 1 - \frac{|T_1+P+U|}{n}$} \\
    & = \frac{1}{512}  \left( \frac{1}{n} \right)^{1 - \frac{|T_1+P+U|}{2n}} \ \NSW(\mathcal{N}) \\
   &\geq \frac{1}{512} \ \left( \frac{1}{n} \right)^{\frac{53}{54}} \ \NSW(\mathcal{N}) \tag{since $|T_1+P+U| \geq \frac{n}{27}$}
\end{align*}
The lemma stands proved.
\end{proof}

\subsection{Properties for Invoking Algorithm \ref{Alg:CapSW}}
\label{subsection:propalgthree}
Complementing the previous subsection, we now consider Case \textcal{2}: $|\overline{U}| \geq \frac{26}{27} n$. Recall that to apply the guarantee obtained for \textsc{CappedSocialWelfare} (i.e., Theorem \ref{thm:cappedsw}), we need the instance at hand to satisfy the following two properties, with some underlying allocation $\mathcal{O} = (O_1, \ldots, O_n)$ and set of agents $\overline{A} \subseteq [n]$:

\begin{itemize}
    \item[] \textbf{P1}: The welfare $\sum_{i\in \overline{A}} \ \widehat{v}_i(O_i) \geq \frac{26}{27} \sqrt{n}$.
    \item[] \textbf{P2}: For each agent $i\in \overline{A}$ and all goods $g'\in O_i$, the value $\widehat{v}_i(g')\leq \frac{1}{2\sqrt{n}}$.
\end{itemize}

In this subsection, we will show that these properties hold for the instance $\instance{[n], {\R'}, \{v_i\}_{i}}$, bundles $O_j = N_j \cap \R'$ (for all agents $j \in [n]$) and subset $\overline{A} = \overline{U}$. This will enable us to instantiate Theorem \ref{thm:cappedsw} in the next subsection. 

\begin{proposition}\label{lem:highsw}
If $|\overline{U}| \geq \frac{26n}{27}$, then (with high probability) the allocation, $\mathcal{O} = (O_1, \ldots, O_n)$, with bundles $O_i = N_i \cap \R'$ (for all $i \in [n]$), and subset $\overline{A} = \overline{U}$ satisfy property \textbf{P1} mentioned above. 
\end{proposition}
\begin{proof}
We will first note that, for each agent $i \in \overline{U}$, with high probability, we have 
\begin{align}
v_i(N_i \cap \R' ) \geq \frac{1}{4} v_i(N_i) \label{ineq:blr}
\end{align}

Recall that $\R'$ is the  subset obtained by randomly partitioning $[m] \setminus \left( \M + \big( \pi([n]) \big) \right)$ (in Algorithm \ref{Alg:NSWXOS}). Also, $\overline{U} \subseteq \overline{P} + T_3$. For agents in $T_3$, one can  obtain inequality (\ref{ineq:blr}) by employing arguments similar to the ones used to establish inequality (\ref{eqn:T3conc}) (for the symmetric case of $\R$). Also,  arguments analogous to the ones used for inequality (\ref{eq:concPbar}), establish inequality (\ref{ineq:blr}) for agents in $\overline{P}$. Therefore, inequality (\ref{ineq:blr}) holds for all agents $i \in \overline{U} \subseteq \overline{P} + T_3$.

Algorithm \ref{Alg:NSWXOS} calls the subroutine \textsc{CappedSocialWelfare} with  parameters $\beta_i = \frac{1}{n} \ \frac{1}{v_i(X_i + \pi(i))}$ (Step \ref{step:beta} in Algorithm \ref{Alg:NSWXOS}). Therefore, using the definition from equation (\ref{eqn:hat-v}), we get 
\begin{align}
\widehat{v}_i(N_i \cap \R') = \min \left\{ \frac{1}{\sqrt{n}}, \ \frac{v_i( N_i \cap \R')}{ n \  v_i(X_i + \pi(i))} \right\} \label{ineq:iisc}
\end{align}
Note that, for each agent $i \in \overline{U}$, by definition, we have $v_i(X_i + \pi(i)) < \frac{1}{4 \sqrt{n}} v_i(N_i)$. This inequality along with equations (\ref{ineq:blr}) and (\ref{ineq:iisc}) lead to the bound $\widehat{v}_i(N_i \cap \R') = \frac{1}{\sqrt{n}}$, for all agents $i \in \overline{U}$. 

Therefore, property \textbf{P1} holds, with bundles $O_j  = N_j \cap \R'$ (for all agents $j \in [n]$) and subset $\overline{A} = \overline{U}$:
\begin{align}\label{eq:lwbound}
\sum_{i \in \overline{A}} \widehat{v}_i ( O_i) = \sum_{i \in \overline{U}} \widehat{v}_i(N_i \cap \R')  = |\overline{U}| \frac{1}{\sqrt{n}} \ {\geq} \ \frac{26 \sqrt{n}}{27} 
\end{align}  
The last inequality follows from $|\overline{U}| \geq \frac{26n}{27}$. The proposition stands proved.
\end{proof}

Next, we will address property \textbf{P2}.

\begin{proposition}\label{lem:smallgoods}
 With high probability, for each agent $i \in \overline{U}$ and each good $g' \in N_j \cap \R'$, we have $\widehat{v}_i(g') < \frac{1}{2\sqrt{n}}$, i.e., property \textbf{P2} holds with bundles $O_j  = N_j \cap \R'$ (for all agents $j \in [n]$) and subset $\overline{A} = \overline{U}$.
\end{proposition}
\begin{proof}
For each agent $i \in \overline{U} \subseteq \overline{P} + T_3 \subseteq T_2 + T_3$, by the definition of these sets, the value of any good $g' \in N_i$ is upper bounded as follows $v_i(g') \leq \frac{1}{256 \sqrt{n}} v_i(N_i)$. 

Furthermore, Lemma \ref{lem:linearT3} and \ref{lem:linearPbar} ensure that (with high probability) $v_i(X_i + \pi(i)) \geq \frac{1}{\constlinear n} v_i(N_i)$, for all agents $i \in \overline{P} + T_3 \supseteq \overline{U}$.

Using these observations and the definition from equation (\ref{eqn:hat-v}), we get, for all agents $i \in \overline{U}$:
\begin{align*}
\widehat{v}_i(g') = \min \left\{ \frac{1}{\sqrt{n}}, \ \frac{v_i(g')}{ n \  v_i(X_i + \pi(i))} \right\} \leq \min \left\{ \frac{1}{\sqrt{n}}, \ \frac{ 64 \ v_i(g')}{ v_i(N_i)} \right\} \leq \min \left\{ \frac{1}{\sqrt{n}}, \ \frac{ 64 }{ 256 \sqrt{n}} \right\} < \frac{1}{2\sqrt{n}}
\end{align*}
This completes the proof. 
\end{proof}

In the next subsection, we will use the Propositions \ref{lem:highsw} and \ref{lem:smallgoods} to invoke Theorem \ref{thm:cappedsw} and, consequently, obtain a sublinear approximation ratio in Case \textcal{2}: $|\overline{U}| \geq \sfrac{26n}{27}$. 


\subsection{Sublinear Approximation Guarantee in Case \textcal{2}}
\label{subsection:casetwo}

\begin{lemma}\label{lem:case2}
If $|\overline{U}| \geq \frac{26n}{27}$, then the Nash social welfare of allocation $\Q^*$ is at least $\frac{c}{n^{\nicefrac{53}{54}}}$ times the optimal Nash social welfare,
\begin{align*}
\NSW(\Q^*) \geq \frac{c}{n^{\nicefrac{53}{54}}}  \NSW(\mathcal{N}).
\end{align*}
Here, $c \in \mathbb{R}_+$ is a fixed constant.
\end{lemma}
\begin{proof}
We will show that, in the current setting  and with high probability, for at least $\frac{1}{27} n$ agents $i$ (in the set $\overline{U}$) for whom we have $v_i(Y_i) \geq \frac{c}{\sqrt{n}} \  v_i(N_i)$; here $Y_i$ is the bundle assigned to agent $i$ in the \textsc{CappedSocialWelfare} subroutine and $c$ is a fixed constant. This $O(\sqrt{n})$ bound for at least $\frac{1}{27}n$ agents, along with the fact that all the remaining agents achieve at a least linear approximation in allocation $\Q^*$ (see Lemmas \ref{lem:sublinearT1} to \ref{lem:linearPbar}), overall gives us a sublinear approximation ratio.

We will start by lower bounding the social welfare (under $\widehat{v}_i$s) that agents in $\overline{U}$ achieve through the \textsc{CappedSocialWelfare} subroutine. Towards this, recall that in the current case (i.e., with $|\overline{U}| \geq \frac{26n}{27}$) Propositions \ref{lem:highsw} and \ref{lem:smallgoods} ensures that, with high probability,  bundles $O_i = N_i \cap \R'$ (for all agents $i \in [n]$) and subset $\overline{A} = \overline{U}$ satisfy the properties \textbf{P1} and \textbf{P2}. Hence, invoking Theorem \ref{thm:cappedsw} over the instance $\instance{[n], \R', \{ v_i\}_{i \in [n]}}$, we obtain 
\begin{align}
\sum_{j =1}^n \widehat{v}_j(Y_j) \geq \frac{2}{25} \sum_{ j \in \overline{A} } \widehat{v}_j(O_j) \underset{\text{(via \textbf{P1})}}{\geq} \frac{52}{675} \sqrt{n} \label{ineq:swover}
\end{align}
Furthermore, note that, by definition, the capped valuations $\widehat{v}_j$s are upper bounded by $\frac{1}{\sqrt{n}}$. Also, in the current case $\left|[n] \setminus \overline{U} \right| \leq \frac{n}{27}$. Therefore, inequality (\ref{ineq:swover}) reduces to
\begin{equation}
\sum_{i\in \overline{U}} \widehat{v}_i(Y_i) \geq \frac{52}{675}\sqrt{n}- \sum_{j \in [n] \setminus \overline{U}} \widehat{v}_j (Y_j) \geq \frac{52}{675}\sqrt{n}-  \left|[n] \setminus \overline{U} \right| \ \frac{1}{\sqrt{n}} \geq \frac{52}{675}\sqrt{n}- \frac{1}{27}\sqrt{n}=\frac{1}{25}\sqrt{n} \label{eq:gamsw}
\end{equation}
	
With the fixed constant $c \coloneqq \frac{1}{2.2 \times 10^4}$, write $B \coloneqq \left\{ i \in \overline{U} \ : \ v_i(Y_i) \geq \frac{c}{\sqrt{n}} \ v_i(N_i) \right\}$. Next, we will use inequality (\ref{eq:gamsw}) to prove that the cardinality of $B$ is at least $\frac{1}{27}n$.
In particular, for all agents $j \in \overline{U} \setminus B$ we have 
\begin{align*}
\widehat{v}_j(Y_j) &= \min \left\{ \frac{1}{\sqrt{n}}, \ \frac{v_j(Y_j)}{n \ v_j(X_j + \pi(j))} \right\} \tag{by defn.~of $\widehat{v}_j$ with $\beta_j = \frac{1}{n v_j(X_j + \pi(j))}$}\\
&\leq \min \left\{ \frac{1}{\sqrt{n}}, \ \frac{c \ v_j(N_j)}{ \sqrt{n}} \cdot \frac{1}{ n \ v_j(X_j + \pi(j))}\right\} \tag{since $j \notin B$} \\
& \leq \min \left\{ \frac{1}{\sqrt{n}}, \ \frac{c \ v_j(N_j)}{ \sqrt{n}} \cdot \frac{64}{ v_j(N_j) }\right\} \tag{via Lemmas \ref{lem:linearT3} \& \ref{lem:linearPbar}; $j \in \overline{U} \subseteq T_3 + \overline{P}$} \\
& = \frac{64 c}{\sqrt{n}} \tag{constant $c$ is sufficiently small}
\end{align*}

Towards a contradiction, assume that $|B| < \frac{1}{27} n$. Then,
\begin{align}
\sum_{i\in \overline{U}}\widehat{v}_i(Y_i) = \sum_{i \in B}  \widehat{v}_i(Y_i) + \sum_{j\in \overline{U} \setminus B} \widehat{v}_j(Y_j) \leq \sum_{i \in B} \frac{1}{\sqrt{n}}+ \sum_{i \in \overline{U} \setminus B} \frac{64 c}{\sqrt{n}} < \frac{n}{27} \ \frac{1}{\sqrt{n}}+ n \  \frac{64 c}{\sqrt{n}} \label{ineq:swcontra}
<\frac{\sqrt{n}}{25}
\end{align}
Here, the last inequality holds since constant $c < \frac{1}{64} \left( \frac{1}{25} - \frac{1}{27} \right)=\frac{1}{21600}$. Given that inequality (\ref{ineq:swcontra}) contradicts the bound (\ref{eq:gamsw}), we necessarily have $|B| \geq \frac{n}{27}$. Also, note that the definition of the set $B \subseteq \overline{U}$ and the monotonicity of $v_i$ imply $v_i(Q^*_i) \geq v_i(Y_i) \geq \frac{c}{\sqrt{n}} \ v_i(N_i) $, for all agents $i \in B$, i.e., 
\begin{align}
\left( \prod_{i \in B} v_i(Q^*_i) \right)^{\frac{1}{n}} \geq \left( \frac{c}{\sqrt{n}} \right)^{\frac{|B|}{n}} \left( \prod_{i \in B} v_i(N_i) \right)^{\frac{1}{n}} \label{ineq:nswB}
\end{align} 
	
Furthermore, for all agents $j \in \left( \overline{U} \setminus B \right) \subseteq \overline{P} + T_3 $, Lemmas \ref{lem:linearT3} and \ref{lem:linearPbar} give us 
\begin{align}
\left( \prod_{j \in \overline{U} \setminus B} v_j(Q^*_j) \right)^{\frac{1}{n}} \geq \left( \frac{1}{64 n} \right)^{\frac{|\overline{U} \setminus B|}{n}} \left( \prod_{j \in \overline{U} \setminus B} v_j(N_j) \right)^{\frac{1}{n}} \label{ineq:nswNB}
\end{align} 
For the remaining agents $ h \in T_1 + P + U  = [n] \setminus \overline{U}$, using Lemmas \ref{lem:sublinearT1}, \ref{lem:sublinearP}, and \ref{lem:sublinearU}, we get  	
\begin{align}
\left( \prod_{h \in [n] \setminus \overline{U}} v_h(Q^*_h) \right)^{\frac{1}{n}} \geq \frac{1}{2}\left( \frac{1}{256 \sqrt{n}} \right)^{\frac{|[n] \setminus \overline{U}|}{n}} \left( \prod_{h \in [n] \setminus \overline{U}} v_h(N_h) \right)^{\frac{1}{n}} \geq \frac{1}{2}\left( \frac{1}{64 {n}} \right)^{\frac{|[n] \setminus \overline{U}|}{n}} \left( \prod_{h \in [n] \setminus \overline{U}} v_h(N_h) \right)^{\frac{1}{n}} \label{ineq:nswT1PU}
\end{align}
Combining inequalities (\ref{ineq:nswB}), (\ref{ineq:nswNB}), and (\ref{ineq:nswT1PU}), we obtain the desired sublinear approximation guarantee
\begin{align*}
\NSW(\Q^*) & = \left( \prod_{i \in B} v_i(Q^*_i) \right)^{\frac{1}{n}}  \left( \prod_{i \in [n] \setminus B} v_i(Q^*_i) \right)^{\frac{1}{n}} \\
& \geq  \left( \frac{c}{\sqrt{n}} \right)^{\frac{|B|}{n}} \frac{1}{2}\left( \frac{1}{64 {n}} \right)^{\frac{|[n] \setminus B|}{n}} \left( \prod_{j \in [n]} v_j(N_j) \right)^{\frac{1}{n}} \\
& \geq  \frac{c}{2} \left( \frac{1}{\sqrt{n}} \right)^{\frac{|B|}{n}} \left( \frac{1}{n} \right)^{\frac{|[n] \setminus B|}{n}} \NSW(\mathcal{N}) \tag{constant $c < \frac{1}{64}$} \\
& = \frac{c}{2} \left( \frac{1}{{n}} \right)^{\frac{|B|}{2n}} \left( \frac{1}{n} \right)^{\frac{|[n] \setminus B|}{n}} \NSW(\mathcal{N}) \\
& = \frac{c}{2} \left( \frac{1}{{n}} \right)^{1 - \frac{|B|}{2n}} \NSW(\mathcal{N}) \tag{$\frac{|[n] \setminus B|}{n} = 1 - \frac{|B|}{n}$} \\
& \geq \frac{c}{2}  \left( \frac{1}{{n}} \right)^{\frac{53}{54}} \NSW(\mathcal{N}) \tag{since $|B| \geq \frac{n}{27}$}
\end{align*}
This completes the analysis for Case \textcal{2} and establishes the lemma.
\end{proof}

\subsection{Proof of Theorem \ref{thm:main}}
\label{subsection:mainproof}

We now restate Theorem \ref{thm:main} and show that it follows from Lemmas \ref{lem:case1} and Lemma \ref{lem:case2}.

\mainthm*
\begin{proof}
Recall that $\Q$ is the allocation returned by Algorithm \ref{Alg:NSWXOS} and $\N$ is a Nash optimal allocation. To prove the theorem, we consider the (previously-mentioned) exhaustive cases:\\

\noindent \textbf{Case \textcal{1}}: $|T_1 + U + P|\geq \frac{n}{27}$. In this case, Lemmas \ref{lem:reduce} and \ref{lem:case1} give us $\NSW (\Q) \geq \frac{1}{2} \NSW(\Q^*) \geq \frac{c'}{2 n^{\sfrac{53}{54}}} \NSW(\N)$.
\newline

\noindent \textbf{Case \textcal{2}}: $|\overline{U}|\geq \frac{26n}{27}$. Here, via Lemmas \ref{lem:reduce} and \ref{lem:case2}, we obtain (with high probability) $\NSW (\Q)\geq   \frac{1}{2} \NSW(\Q^*) \geq \frac{c}{2 n^{\sfrac{53}{54}}} \NSW(\N)$.

This case analysis shows that overall Algorithm \ref{Alg:NSWXOS} achieves an approximation ratio of ${O}(n^{\sfrac{53}{54}})$ for the Nash social welfare maximization problem, under $\XOS$ valuations. The theorem stands proved. 
\end{proof}

\paragraph{Remark.}
Here we address the corner case wherein for some agents $z \in [n]$ the value of the good assigned in Step \ref{step:pi} of Algorithm \ref{Alg:NSWXOS}  is zero, i.e., $v_z(\pi(z)) = 0$. Note that for remaining agents $i$, we have $v_i(\pi(i)) >0$ and, hence, $v_i(X_i + \pi(i)) >0$. Consequently, for such agents $i$, the parameter $\beta_i$ (considered in Step \ref{step:beta} of the algorithm) is well-defined.  

Note that, if for an agent $z$ we have $v_z(\pi(z)) = 0$, then $z$ necessarily belongs to either set $T_1$ or set $P$. This follows from the observation that, for such an agent $z$, all the goods in the set $[m] \setminus \big( \M + \pi([n]) \big)$ are of zero value. Equivalently, all the goods of nonzero value for $z$ are contained in $\big( \M + \pi([n]) \big)$. Therefore, $z$ has nonzero value for at most $n \log n + n \leq 2n \log n$ goods. For an $\XOS$ (subadditive) valuation $v_z$, this implies that $v_z(g_z^*)$ cannot be less than $\frac{1}{16n\log n} v_z(N_z)$ and, hence, $z \notin T_3$. Furthermore, the fact that $v_z \left( [m] \setminus \big( \M + \pi([n]) \big) \right) = 0$ gives us $z \notin \overline{P}$. Therefore, $z$ must be contained in $T_1 \cup P$. 

We exclude such agents $z$ from phases three and four of Algorithm \ref{Alg:NSWXOS}; this ensures that we do not have to consider $\beta_z$. Then, such agents $z$ are directly considered in Step \ref{step:mu} with $X_z = Y_z = \emptyset$. Since $z \in T_1 \cup P$, the arguments from Lemmas \ref{lem:sublinearT1} and \ref{lem:sublinearP} provide a sublinear guarantee for $z$ even with $X_z = Y_z = \emptyset$. 

For the remaining agents $i$ (with the property that $v_i(\pi(i)) >0$), the guarantees obtained for phases three and four (in particular, the ones obtained in Lemmas \ref{lem:linearT3} and \ref{lem:linearPbar}) in fact improve, since the number of agents under consideration gets reduced. These observations imply that the sublinear approximation guarantee holds as is in Case \textcal{1}. For Case \textcal{2} (i.e., when $|\overline{U}| > \sfrac{26n}{27}$), note that, $\overline{U} \cap (T_1 \cup P) = \emptyset$. Therefore, Lemma \ref{lem:case2} is applicable and we obtain the stated approximation ratio throughout.

\section{Query Lower Bound for $\NSW$}
\label{appendix:comm-comp-nsw}

This section shows that, under $\XOS$ valuations, exponentially many queries are necessarily required to approximate $\NSW$ within a factor of $\left(1 - \frac{1}{e}\right)$ (Theorem \ref{thm:xoshard}). We establish this lower bound by considering a communication complexity setup in which each agent $i$ holds her $\XOS$ valuation $v_i$.  We will show that in this setup approximating $\NSW$, within a factor of $\left(1 - \frac{1}{e}\right)$, requires exponential communication among the agents. Note that any approximation algorithm---that requires a sub-exponential number of demand and $\XOS$ queries---directly translates into a protocol that achieves the same approximation guarantee with sub-exponential communication. Hence, contrapositively, the communication lower bound mentioned above will prove the desired hardness result. Formally, 

\xoshard*

To prove this theorem we use the multi-disjointness problem; see, e.g., \cite{roughgarden2016communication}. The input to this problem is a set of agents $[n]$ and a set of elements $[t]$. Each agent $i$ holds a subset $B_i \subseteq \{1, 2, \ldots, t \}$ and the problem is to distinguish between the following two extreme cases: \\

\noindent
\textbf{Case 1:} {\it Totally Intersecting.} In this case, there is an element $b \in [t]$ that belongs in the subset of each agent $i \in [n]$, i.e., $b \in \cap_{i=1}^n B_i$. 

\noindent
\textbf{Case 2:} {\it Totally Disjoint.} In this case, the subsets held by agents are pairwise disjoint, $B_i \cap B_j = \emptyset$ for all $i, j \in [n]$. \\

It is known that any protocol (including randomized ones with two sided errors) that distinguishes between the above two cases requires $\Omega(\sfrac{t}{n})$ communication; see~\cite{roughgarden2016communication} and references therein.\footnote{Note that, for multi-disjointness problem instances that do not fall under the two cases, the output of the protocol can be arbitrary.} 

We will invoke multi-disjointness with $t$ exponentially larger than $n$. A similar reduction from multi-disjointness was used in \cite{dobzinski2010approximation} for establishing the communication complexity of social welfare maximization. 

The reduction is based on a combinatorial gadget (see Definition \ref{defn:covering}), whose existence will be established via the probabilistic method (Lemma \ref{lemma:equicover-exist}). We define an $n$-equipartition $\mathcal{P}=(P_1, \ldots, P_n)$ of a set $[m]$ as an ordered collection of $n$ equal sized subsets, $P_1, \ldots, P_n$, that satisfy $\cup_{j=1}^n P_j = [m]$ and $P_k \cap P_j = \emptyset$, for all $k \neq j$ (i.e., the subsets are pairwise disjoint). In the following definition all the considered equipartitions $\mathcal{P}^s = (P^s_1, \ldots, P^s_i, \ldots, P^s_n)$ are of the set $[m]$.    

\begin{definition}[$(n, r, \varepsilon)$-equicovering]
\label{defn:covering}
A family of $r$ equipartitions $\{ \mathcal{P}^s \}_{s=1}^r$  is said to be an $(n, r, \varepsilon)$-equicovering (of $[m]$) iff for all indices $s_1, s_2, \ldots, s_n \in [r]$, such that no two are equal ($s_x \neq s_y$ for all $x \neq y$), we have
\begin{align*}
    \left| \bigcup_{i \in [n]} P^{s_i}_i \right| \leq m \left(1 - \left(1 - \frac{1}{n}\right)^n + \varepsilon \right).
\end{align*}
Here, $P^{s_i}_i$ is the $i^{th}$ subset in the $s_{i}^{th}$ equipartition $\mathcal{P}^{s_i}=(P^{s_i}_1, \ldots, P^{s_i}_n)$.
\end{definition}

The lemma below establishes existence of $(n,r, \varepsilon)$-equicoverings with exponentially large $r$. 
 
\begin{lemma}
\label{lemma:equicover-exist}
There exists an $(n, r, \varepsilon)$-equicovering $\{ \mathcal{P}^s \}_{s=1}^r$ (of $[m]$) in which the number of equipartitions $r \geq e^{\left(m \varepsilon^2\right)/n}$. 
\end{lemma}
\begin{proof}
We use the probabilistic method to prove this lemma. For parameter $r \in \mathbb{Z}_+$, to be explicitly set later, we construct a family of $r$ equipartitions as follows: for each $s \in [r]$,  independently set $\mathcal{P}^s = (P^s_1, \ldots, P^s_n)$ to be a uniformly at random $n$-equipartition of $[m]$. We will show that the family $\left\{\mathcal{P}^s\right\}_{s=1}^r$ is an $(n, r, \varepsilon)$-equicovering, with nonzero probability. Hence, via the probabilistic method, the lemma follows. 

Towards establishing the equicovering property, fix any $n$ distinct indices $s_1, s_2, \ldots, s_n \in [r]$ (i.e., $s_x \neq s_y$ for all $x \neq y$). 
For each $g \in [m]$, define an indicator random variable $Y_g = 1$ if $g \in \cup_{i \in [n]} P^{s_i}_i$, and $Y_g = 0$ otherwise. Note that $\E \left[ \left| \bigcup_{i \in [n]} P^{s_i}_i \right| \right] = \E \left[ \sum_{g=1}^m Y_g \right]$. 

In addition, since the equipartitions are selected independently, for each $g \in [m]$ we have 
\begin{align*}
    \Pr \left\{ Y_g = 1\right\} &= \Pr \left\{ g \in \cup_{i \in [n]} P^i_{s_i} \right\}  = 1 - \left( 1 - \frac{1}{n} \right)^n
\end{align*}

Moreover, it is relevant to note random variables $\{ Y_g \}_{g \in [m]}$ are negatively associated \cite{dubhashi1996balls}. One can show this by first considering, for each $g$ and $s_i$, the indicator random variable $Y_g^{s_i} = 1$ iff $g \in P^{s_i}_i$. Here, for any $s_i$ and $g \neq h$, the variables $Y_g^{s_i}$ and $Y_h^{s_i}$ are negatively associated, since the subsets are size restricted. Hence, composing the variables $Y_g^{s_1},\ldots, Y_g^{s_n}$, across all $g$, we obtain that the the random variables $Y_g$ are negatively associated themselves; see, e.g., Proposition 7 in \cite{dubhashi1996balls}.
 
Therefore, given that $\{ Y_g \}_{g \in [m]}$ are negatively associated, the Chernoff-Hoeffding bound implies  
\begin{align*}
    \Pr\left\{ \sum_{g \in [m]} Y_g \geq m \left(1 - \left(1 - \frac{1}{n} \right)^n  + \varepsilon\right)\right\} & \leq e^{-2m\varepsilon^2} 
\end{align*}
Hence, 
\begin{align}
\Pr \left\{ \left| \cup_{i \in [n]} P^{s_i}_i \right|  \geq m \left(1 - \left(1 - \frac{1}{n} \right)^n  + \varepsilon\right) \right\} \leq e^{-2m\varepsilon^2} \label{ineq:chernoff}
\end{align}
The concentration bound (\ref{ineq:chernoff}) holds for any fixed choice of the indices $s_1, \ldots, s_n$. With $r$ different equipartitions in the family $\{\mathcal{P}^s \}_{s=1}^r$, there are at most $r^n$ different choices for the indices $s_1, \ldots, s_n \in [r]$. Therefore, using the union bound, we get that for any $r$ that satisfies $r^n < e^{2m \varepsilon^2}$ an $(n,r,\varepsilon)$-equicovering of $[m]$ is guaranteed to exist. 

Setting $m$ to be polynomially larger than $n$ (say, $m = n^3$), we obtain the existence of $(n,r, \varepsilon)$-equicoverings with $r \geq e^{\left(m \varepsilon^2\right)/n}$. The lemma stands proved. 
\end{proof}

\subsection{Proof of Theorem \ref{thm:xoshard}}
Using equicoverings as gadgets, we will reduce the multi-disjointness problem to $\NSW$ maximization. In particular, we will consider multi-disjointness problem instances in which $t$ is exponentially larger than $n$. In the reduction, $t$ will bound the number of additive functions that define the $\XOS$ valuations of the agents and, hence, it can be exponentially large. Also, recall that, in a multi-disjointness instance, $B_i \subseteq [t]$ denotes the subset held by agent $i$. 

Lemma \ref{lemma:equicover-exist} ensures that for an exponentially large $t$, and a polynomially large (in the number of agents) $m$, an $(n, t, \varepsilon)$-equicovering, $\left\{ \mathcal{P}^s \right\}_{s=1}^t$, exists. Using the equicovering $\left\{ \mathcal{P}^s \right\}_{s=1}^t$ we define the $\XOS$ valuations of the $n$ agents in the fair division instance. Specifically, for each agent $i \in [n]$, we populate a family $\mathcal{F}_i$ of additive functions such that the valuation $v_i(S) \coloneqq \max_{f \in \mathcal{F}_i} f(S)$, for all subsets of goods $S \subseteq [m]$. The family $\mathcal{F}_i$ is obtained as follows: 
\begin{itemize}
\item For every element $s \in B_i$, consider the subset $P^s_i$, i.e., consider the $i^{th}$ subset in the $s^{th}$ equipartition $\mathcal{P}^s = (P^s_1, \ldots, P^s_n)$. 
\item Define additive function $f^{(s)}(\cdot)$ as the indicator function of subset $P^s_i$. In particular, $f^{(s)} (g) = 1$, for each $g \in P^s_i$, and $f(h) = 0$ for all $h \notin P^s_i$.  
\item For each $s \in B_i$, include $f^{(s)}$ in $\mathcal{F}_i$.
\end{itemize}

Now, if the given multi-disjointness instance satisfies Case $1$ (i.e., it is totally intersecting), then there exists an element $x \in B_i$ for all $i \in [n]$. Therefore, we can assign to each agent $i \in [n]$ the set of goods $P^x_i$. Since $\mathcal{P}^x=(P^x_1, \ldots, P^x_n)$ is an equipartition, we get that, in this case, the $\NSW$ is $\sfrac{m}{n}$. 

On the other hand, if the multi-disjointness instance falls under Case $2$ (i.e., it is totally disjoint), then by the property of $(n, t, \varepsilon)$-equicovering, we get that the social welfare for \emph{any} allocation of goods is at most $m\left(1 - \left(1 - \frac{1}{n}\right)^n + \varepsilon \right)$. In particular, for any allocation $(A_1, \ldots, A_n)$ of the goods write $f_{i, A_i}(\cdot)$ to denote the additive function (from $\mathcal{F}_i$) that induces $v_i(A_i)$. The additive function $f_{i, A_i}$ must have been included in $\mathcal{F}_i$ considering some element $s_i \in B_i$. In fact, $f_{i, A_i}$ is the indicator function of the subset $P^{s_i}_i$. Therefore, $v_i(A_i) = f_{i, A_i} (A_i) = \left|A_i \cap P^{s_i}_i \right|$. Since $(A_1, \ldots, A_n)$ is a partition of $[m]$,  the subsets $\left( A_i \cap P^{s_i}_i \right)$---for $i \in [n]$---are pairwise disjoint and their union is contained in $\cup_{i = 1}^n P^{s_i}_i$. These observations lead to the stated bound on social welfare: 
\begin{align*}
\sum_{i=1}^n v_i(A_i) & = \sum_{i=1}^n \left|A_i \cap P^{s_i}_i \right| \\ 
& \leq \left| \cup_{i = 1}^n P^{s_i}_i \right| \\
& \leq m\left(1 - \left(1 - \frac{1}{n}\right)^n + \varepsilon \right) \tag{since $\left\{ \mathcal{P}^s \right\}_{s=1}^t$ is an $(n,t, \varepsilon)$-equicover}
\end{align*} 
Hence, in Case $2$, the average social welfare is at most $\frac{m}{n}\left(1 - \left(1 - \frac{1}{n}\right)^n + \varepsilon \right)$. Furthermore, we get that the optimal $\NSW$ in this case is upper bounded by $\frac{m}{n}\left(1 - \left(1 - \frac{1}{n}\right)^n + \varepsilon \right)$. 

The multiplicative gap between the optimal $\NSW$ in the two cases implies that any algorithm that approximates $\NSW$ within a factor of $\left( 1 - \frac{1}{e} + \varepsilon \right)$ can be used to distinguish between the cases. Hence, such an algorithm necessarily requires $\Omega(\sfrac{t}{n})$ communication. In our construction, parameter $t$ is exponentially large. Therefore, exponentially many queries are required for finding an allocation with $\NSW$ at least $\left(1 - \frac{1}{e} + \varepsilon\right)$ times the optimal. The theorem stands proved. 
\section{Conclusion and Future Work}
This work breaks the $O(n)$ approximation barrier that holds for $\NSW$ maximization under $\XOS$ valuations in the value-oracle model. In particular, using demand and $\XOS$ oracles, we obtain the first sublinear approximation algorithm for maximizing $\NSW$ with $\XOS$ valuations. A key innovative contribution of the work is the connection established between $\NSW$ and capped social welfare. This connection builds upon the following high-level idea: to achieve a sublinear approximation for $\NSW$, it suffices first to obtain an $O(n) $-approximation guarantee and, then, ensure that a constant fraction of the agents additionally achieve a sublinear approximation of their optimal valuation. 

Understanding the limitations of demand queries---in the $\NSW$ context---is a relevant direction for future work. While we rule out a $\left(1 - 1/e \right)$-approximation with subexponentially many queries and it is known that $\NSW$ maximization is ${\rm APX}$-hard under demand queries (see, e.g., \cite{barman2021approximating}), it would be interesting to obtain stronger inapproximability results. Developing a sublinear approximation guarantee for subadditive valuations will also be interesting.

\newpage 
\appendix

\section{Missing Proofs from Section \ref{sec:phases}}

\subsection{Proof of Lemma \ref{lemma:matchhigh}}
\label{appendix:repeated-mathcings}

In this section, we restate and prove Lemma \ref{lemma:matchhigh}.

\LemmaRepeatedMatchings*
\begin{proof}
Write $\tau_1,\tau_2,\ldots, \tau_{\log{n}}$ to denote the maximum-product matchings that constitute $\M$, i.e., $\M=\bigcup\limits_{k=1}^{\log{n}}\tau_k([n])$. Here, $\tau_{k}([n])$ denotes the set of goods matched in $k^{\text{th}}$ iteration of the for-loop, $\tau_{k}([n]) \coloneqq \{ \tau_{k} (i) \}_{i\in [n]}$.

Initialize subset of (``satisfied'') agents $S_0=\emptyset$. Iteratively considering $\tau_1,\tau_2,\ldots,\tau_{\log n}$, we will identify subsets $S_0\subseteq S_1\subseteq\ldots \subseteq S_{\log{n}} = [n]$ such that for all indices $k\in [\log{n}]$ and agents $i\in S_k$ we have a {\it distinct} good $h(i)$ with $v_i(h(i))\geq v_i(g^*_i)$. That is, for every index $k$, each agent in $S_k$ will have a distinct good assigned to her with the desired value. We will also maintain the invariant that for each index $k$ and every (unsatisfied) agent $j\in [n]\setminus S_k$ the good $g^*_j \notin \cup_{\ell=1}^k\tau_{\ell}([n])$. 

Towards maintaining the invariants, iteratively for each $k\in \{1,2,\ldots, \log{n}\}$, consider the matching $\tau_k$ and for every unsatisfied agent after the previous iteration, $j\in [n]\setminus S_{k-1}$ with the property that $g^*_j\in \tau_k([n])$, set $h(j)=g^*_j \in \tau_k([n]) \subseteq \M$. Write $D'_k$ to denote the subset of these agents, $D'_k \coloneqq \left\{j\in [n]\setminus S_{k-1} : g^*_j \in \tau_k([n]) \right\}$. In addition, consider each (unsatisfied) agent $a\in [n]\setminus ( S_{k-1} + D'_k) $ that is matched to a good $\tau_k(a)\notin {\{g^*_j\}}_{j\in[n]\setminus S_{k-1}}$, i.e., the good $\tau_k(a)$ is not an optimal good for any other \emph{unsatisfied} agent. Let $D_k''$ denote these agents, $D^{''}_k \coloneqq \left\{ a \in [n] \setminus (S_{k-1} + D'_k) : \tau_k(a)\notin \{g^*_j\}_{j\in[n] \setminus S_{k-1}} \right\}$. For all $a\in D^{''}_k$ set $h(a)=\tau_k(a)$. 
	
Now, set $S_k = S_{k-1}+D^{'}_k +D^{''}_k$. Note that every agent $j\in S_k\setminus S_{k-1}$ is assigned a distinct good $h(j)$ from $\tau_k([n])$. Hence, the assigned goods overall constitute a matching. Also, inductively we have that for every agent $j\in [n]\setminus S_{k-1}$, good $g^*_j \notin \cup_{\ell=1}^{k} \tau_{\ell}([n])$. The construction of $D^{'}_k$ ensures that this property continues to hold for $[n]\setminus \left(S_{k-1}+D^{'}_k\right)$, and hence for $[n]\setminus S_k$. In particular, for all $a\in D^{''}_k \subseteq [n]\setminus (S_{k-1}+D^{'}_k)$, we have $g^{*}_a \notin \cup_{\ell=1}^k \tau_{\ell}([n])$. Now, since $\tau_k$ is an optimal matching over the remaining goods, we have that $v_a(\tau_k(a))\geq v_a(g^{*}_a)$ for all $a\in D^{''}_k$. Therefore, for all agents $j\in S_k\setminus S_{k-1} = D^{'}_k + D^{''}_k$ we have 
	\begin{align*}
		v_j(h(j))\geq v_j(g^{*}_j)
	\end{align*}
	
These observations show that for each $k $ we can extend the matching $h$ such that every agent $i\in S_k$ is assigned a good, $h(i)$, of value at least $v_i(g^{*}_i)$. It remains to show that we need to consider at most $\log{n}$ matchings. 

Towards this, we next prove that, the number of unsatisfied agents decreases by at least a factor of $\frac{1}{2}$ after considering each $\tau_k$; in particular, $|[n]\setminus S_k|\leq \frac{1}{2}|[n]\setminus S_{k-1}|$ for all $k$. Hence, after $\log{n}$ matchings all the agents are satisfied, $S_{\log{n}}=[n]$. 	
Recall that $S_k=S_{k-1}+D^{'}_k+D^{''}_k$. Hence, 
	\begin{align}\label{eq:S_k}
		\left([n]\setminus S_k\right)+D^{'}_k+D^{''}_k = [n]\setminus S_{k-1}
	\end{align}
	
Consider each remaining unsatisfied agent $b\in [n]\setminus S_k$. The fact that $b\notin D^{''}_k$ implies that $\tau_k(b) = g^{*}_x$ for a specific agent $x\in D^{'}_k$. Therefore, $|[n]\setminus S_k|\leq |D^{'}_k|$. This inequality along with equation $\eqref{eq:S_k}$ leads to the desired bound $|[n]\setminus S_k| \leq \frac{1}{2} |[n]\setminus S_{k-1}|$. 
\end{proof}

\subsection{Proof of Lemma \ref{lemma:concentration}}
\label{appendix:concentration}

This section restates and establishes Lemma \ref{lemma:concentration}.

\LemmaConcentration*
\begin{proof}
Write $s$ to denote the cardinality of the subset $\oveN_i$, i.e., $s \coloneqq |\oveN_i|$ and assume, by reindexing, that $\oveN_i = \{1, 2, \ldots, s \}$. 
Also, let $f$ be an additive function that induces $v_i(\oveN_i)$, i.e., $f$ is contained in the family of additive functions that defines the $\XOS$ valuation $v_i$ and $v_i(\oveN_i) = f(\oveN_i) = \sum_{g \in \oveN_i} f(g)$.  

Furthermore, write ${\chi}= (\chi_1, \ldots, \chi_s)$ to denote the (random) characteristic vector for the set $\oveN_i \cap \R$; specifically, for each $g \in \oveN_i$ we have 
\begin{align*}
\chi_g = \begin{cases} 1 & \quad \text{ if } g \in \R \\ 
0 & \quad \text{ otherwise, if } g \notin \R.
\end{cases}
\end{align*}
Overloading notation, we will use $f(\chi_1, \ldots, \chi_s)$ to denote the random variable $f(\oveN_i \cap \R) = \sum_{g \in \oveN_i \cap R} \ f(g)$. Since the function $f$ is additive, $f(\chi_1, \ldots, \chi_s) = \sum_{g \in \oveN_i}  \chi_g \ f(g)$. The random selection of $\R$ ensures that, for each $g \in \oveN_i$, we have $\chi_g = 1$ with probability $1/2$, otherwise $\chi_g = 0$. Therefore, the expected value 
\begin{align}
\E \left[ f(\chi_1, \ldots, \chi_s) \right] = \sum_{g \in \oveN_i}  \frac{1}{2} f(g) = \frac{1}{2} f(\oveN_i) = \frac{1}{2} v_i(\oveN_i) \label{ineq:expvalf}
\end{align}

Furthermore, $f(\chi_1, \ldots, \chi_s)$ is Lipschitz in the sense that for any $(\chi_1, \ldots, \chi_s)$ and any $\chi'_g \in \{0,1\}$, the change in function value (under inclusion or exclusion of good $g$) satisfies   
\begin{align}\label{eqn:elmtbound}
|f(\chi_1,\ldots, \chi_g, \ldots, \chi_s) - f(\chi_1,\ldots, \chi'_g, \ldots, \chi_s)| & \leq f(g) 
\end{align}

Applying McDiarmid's inequality, with inequality (\ref{eqn:elmtbound}), we obtain, for all $t \geq 0$:
\begin{align*}
\Pr \left\{ f(\chi_1,\ldots, \chi_s) \leq \allowbreak \E \left[ f \right] - t \right\} \leq \exp \left(-\frac{2t^2}{\sum_{g=1}^s f(g)^2} \right)
\end{align*}
Since $\E \left[ f \right] = \frac{1}{2} v_i(\oveN_i)$ (see equation (\ref{ineq:expvalf})), the previous inequality reduces to 
\begin{align}
 \Pr \left\{ f(\chi_1,\ldots, \chi_s) \leq  \frac{1}{2} v_i(\oveN_i) - t \right\} \leq \exp \left(-\frac{2t^2}{\sum_{g=1}^s f(g)^2} \right) \label{ineq:probinterim}
\end{align}
Here, $f(g)$ is the (deterministic) value of good $g \in \oveN_i$ under the additive function $f$. Also, the fact that $v_i$ is $\XOS$ and the lemma assumption give us  
\begin{align*}
f(g) & \leq v_i(g)\leq \frac{1}{\sqrt{n}}v_i(\oveN_i) \quad \text{ for all } g \in \oveN_i.
\end{align*}
Therefore, 
\begin{align*}
    \nonumber \sum_{g = 1}^s f(g)^2  &\leq \sum_{g=1}^s\frac{v_i(\oveN_i)}{\sqrt{n}}f(g) =\frac{v_i(\oveN_i)}{\sqrt{n}}\sum_{g=1}^s f(g) \\ 
    & =\frac{(v_i(\oveN_i))^2}{\sqrt{n}} \tag{since $\sum_{g \in \oveN_i} f(g) =  f(\oveN_i) = v_i(\oveN_i)$} \\ 
\end{align*}
Using this inequality and equation (\ref{ineq:probinterim}), with $t=\frac{1}{6} v_i(\oveN_i)$, we get 
\begin{align}
\Pr \left\{ f(\chi_1,\ldots, \chi_s) \leq  \frac{1}{3} v_i(\oveN_i) \right\} \leq \exp \left(-\frac{2 \left( v_i(\oveN_i) \right)^2 \ \sqrt{n}}{36 \ \left( v_i(\oveN_i) \right)^2 } \right) = \exp \left(-\frac{ \sqrt{n}}{18} \right) \label{ineq:probend}
\end{align}

By definition, for every realization of $(\chi_1, \ldots, \chi_s)$, we have $f(\chi_1, \ldots, \chi_s) = f(\oveN_i \cap \R)$. Also, recall that $v_i$ is $\XOS$ and $f$ is an additive function contained in the defining family of $v_i$. Hence, $v_i(\oveN_i \cap \R) \geq f(\oveN_i \cap \R) = f(\chi_1, \ldots, \chi_s)$, for all possible realizations of $\R$. This observation and inequality (\ref{ineq:probend}) lead to the desired bound   
\begin{align*}
    \Pr \left\{ v_i(\oveN_i \cap \R) \leq \frac{1}{3} v_i(\oveN_i) \right\} &\leq  \exp \left( -\frac{\sqrt{n}}{18}\right).
\end{align*}

Recall that the goods in $\G$ are included in $\R$ independently with probability $\nicefrac{1}{2}$. Therefore, a symmetric analysis holds for the complement $\R' = \G - \R$. This completes the proof.  
\end{proof}

\subsection{Proof of Lemma \ref{lem:MK}}
\label{appendix:DMK}

Here we prove Lemma \ref{lem:MK}

\LemmaDiscreteMovingKnife*
\begin{proof}
The termination condition of the while-loop in Step \ref{step:endwhileG_i} (of the \textsc{DiscreteMovingKnife} subroutine) ensures that, for each agent $j \in [n]$, the set $G_j$ (obtained at the end of the loop) satisfies $v_j(g) < \frac{1}{16 n} v_j(G_j)$ for all $g \in G_j$. Hence, for the constructed valuations $v'_j$ (Step \ref{step:definevprime}), we have, for all $g \in \R$:
\begin{align}
v'_j(g) < \frac{1}{16n} v'_j(G_j ) = \frac{1}{16n} v'_j(\R) \label{ineq:smallval}
\end{align}
That is, under the valuations $v'_j$s, all the goods are of value less than $\frac{1}{16n}$ times the value of the grand bundle $\R$. Using this property, we will first show that, for all agents $j \in [n]$, the computed allocation $(X_1,\ldots,X_{n})$ satisfies
\begin{align}
v'_j(X_j) \geq \frac{1}{16n} v'_j(G_j)  \label{ineq:forvprime}
\end{align} 
Subsequently, focusing on the set of agents $T= \left\{ i \in [n]  \ : \ \max_{g \in \widetilde{N}_i} \ v_i(g) < \frac{1}{16n} v_i(\widetilde{N}_i) \right\}$, we will establish the containment $G_i \supseteq \widetilde{N}_i$ for all $i \in T$; see Claim \ref{claim:Gsupset} below. This containment and inequality (\ref{ineq:forvprime}) will lead to the desired inequality for all agents $i \in T$:
\begin{align*}
v_i(X_i) & \geq v'_i(X_i) \tag{by definition of $v'_i$} \\
& \geq \frac{1}{16n} v'_i(G_i)  \tag{via (\ref{ineq:forvprime})} \\
& = \frac{1}{16 n} v_i(G_i) \tag{by definition of $v'_i$} \\
& \geq \frac{1}{16 n} v_i (\widetilde{N}_i) \tag{since $G_i \supseteq \widetilde{N}_i$ and $v_i$ is monotonic} 
\end{align*}

Therefore, to complete the proof we will prove inequality (\ref{ineq:forvprime}) and the containment $G_i \supseteq \widetilde{N}_i$, for all $i \in T$.

Towards establishing inequality (\ref{ineq:forvprime}), consider, for any integer $k \in \mathbb{Z}_+$, the $(k+1)^{\text{th}}$ iteration of the while-loop in Step \ref{step:loopESAMK} of the subroutine. At the point in the subroutine's execution, $k$ agents have been assigned a bundle and, hence, the number of remaining agents $|A| = n-k$. Note that if an agent $a$ has already received a bundle (i.e., $a \in [n] \setminus A$), then inequality (\ref{ineq:forvprime}) holds for $a$; see the selection criterion in Step \ref{step:ifcondt} and recall that $v'_a(G_a) = v'_a(\R)$. Moreover, we will show that for each remaining agent (i.e., for each agent $j \in A$) the value of the remaining goods $v'_j(\Gamma) \geq \left(1 - \frac{k}{8n} \right)v'_j(\R)$. Therefore, the while-loop in Step \ref{step:loopESAMK} will continue to execute till all the $n$ agents have received a bundle of value at least $\frac{1}{16n} v'_j(\R) = \frac{1}{16n} v'_j(G_j)$. 

For analyzing the $(k+1)^{\text{th}}$ iteration, fix an agent $j \in A$ and let $P$ be a subset allocated in one of the previous iterations. Note that $v'_j$ is a subadditive function\footnote{Given that $v_j$ is $\XOS$ (i.e., subadditive), for all subsets $S$ and $T$, we have $v'_j(S \cup T) = v_j((S \cup T) \cap G_j) = v_j((S \cap G_j) \cup (T \cap G_j)) \leq v_j(S \cap G_j) + v_j(T \cap G_j) = v'_j(S) + v'_j(T)$.} and, hence, $v'_j(P) \leq v'_j(P \setminus \{g\}) + v'_j(g)$, in particular for the last good $g$ included in $P$ before it was allocated to some agent. Since $P \setminus \{g\}$ was not allocated to any agent (specifically not allocated to $j$), we have $v'_j(P \setminus \{g\}) < \frac{1}{16n} v'_j(\R)$. This inequality and the bound (\ref{ineq:smallval}) give us $v'_j(P) \leq \frac{1}{8n} v'_j(\R)$. Therefore, for agent $j$, the bundles assigned in the previous $k$ iterations were individually of value at most $\frac{1}{8n} v'_j(\R)$. The set of goods assigned so far is $\R \setminus \Gamma$. Using these observations and the subadditivity of $v'_j$ we obtain $v'_j(\R \setminus \Gamma) \leq \frac{k}{8n} v'_j(\R)$. Therefore, the value of the remaining goods $v'_j(\Gamma) \geq v'_j(\R) - v'_j(\R \setminus \Gamma) \geq \left(1 - \frac{k}{8n} \right)v'_j(\R)$; here the first inequality follows from the subadditivity of $v'_j$. Therefore, the agents in $A$ continue to have sufficiently high value with the unassigned set of goods $\Gamma$ and the subroutine allocates bundles, $X_j$s, that satisfy inequality (\ref{ineq:forvprime}). 

While inequality (\ref{ineq:forvprime}) is satisfied by all agents $j \in [n]$, the following claim holds for agents in $T$. 
\begin{claim} \label{claim:Gsupset}
For each agent $i\in T$, the set $G_i$ obtained via the while-loop (in Steps \ref{step:stopcondition} to \ref{step:endwhileG_i}) is a superset of $\widetilde{N}_i$, i.e., $G_i \supseteq \widetilde{N}_i$.
\end{claim}
	\begin{proof}
	Fix an agent $i\in T$. We will show that the containment $G_i  \supseteq \widetilde{N}_i$ inductively continues to holds through the execution of the while loop. At initialization, we have $G_i = \R$ and, hence, the containment holds. Assume, for induction, that $G_i  \supseteq \widetilde{N}_i$ until the $\ell^{\text{th}}$ iteration of the while-loop and write $g'$ to denote the good removed from $G_i$ in the $(\ell+1)^{\text{th}}$ iteration. For this good we have $v_i(g') \geq \frac{1}{16n} v_i(G_i) \geq \frac{1}{16n} v_i(\widetilde{N}_i)$; the last inequality follows from the induction hypothesis. However, all the goods $g \in \widetilde{N}_i$ satisfy $v_i(g) < \frac{1}{16n} v_i(\widetilde{N}_i)$; recall that $i \in T$. Hence, $g' \notin \widetilde{N}_i$ and $(G_i \setminus \{ g' \}) \supseteq \widetilde{N}_i$. That is, the containment continues to hold and the claim follows. 
\end{proof}
As mentioned previously, Claim \ref{claim:Gsupset} and inequality (\ref{ineq:forvprime}) lead to the desired bound $v_i(X_i)\geq \frac{1}{\constMK n} v_i(\widetilde{N}_i)$, for all agents $i \in T$. The lemma stands proved.  
\end{proof}

\section{Missing Proof from Section \ref{sec:analysis}}
\label{appendix:h-to-qstar}

This section provides a proof of Lemma \ref{lem:reduce}. 

\LemmaReduce*
\begin{proof}
Recall that $\mu$ is a matching between the set of agents, $[n]$, and the goods $\M$ separated in Phase {\rm I} of Algorithm $\ref{Alg:NSWXOS}$ (Step \ref{step:mu} of Algorithm \ref{Alg:NSWXOS}). Also, Lemma \ref{lemma:matchhigh} ensures that there exists a matching $h: [n] \mapsto \M$ such that $v_i(h(i)) \geq v_i(g_i^*)$ for all agents $i \in [n]$. Note that $\mu$ is a maximum-product matching (with offset $\pi(i) + X_i + Y_i$) and $h$ is a feasible matching. Hence,  
\begin{align} 
\NSW(\Q) &= \left( \prod_{i=1}^n v_i \left( \mu(i) + \pi(i) + X_i + Y_i \right) \right)^{\sfrac{1}{n}} \nonumber \\
&\geq \left(\prod_{i=1}^n v_i \left(h(i) + \pi(i) + X_i + Y_i \right) \right)^{\sfrac{1}{n}} \label{eqn:NSWQ}
\end{align}

Furthermore, the monotonicity of $v_i$ gives us $v_i(h(i) + \pi(i) + X_i + Y_i) \geq v_i(h(i))$ and $v_i(h(i) + \pi(i) + X_i + Y_i) \geq v_i(\pi(i) + X_i + Y_i)$. Therefore, $v_i(h_i + \pi(i) + X_i + Y_i) \geq \frac{1}{2}\big( v_i(h_i) \ + \ v_i(\pi(i) + X_i + Y_i) \big)$.
Substituting this inequality into equation (\ref{eqn:NSWQ}), we get
\begin{align*}
    \NSW(\Q) &\geq  \left( \prod_{i =1}^n \frac{1}{2} \big( v_i(h_i) \ + \ v_i(\pi(i) + X_i + Y_i) \big)\right)^{\sfrac{1}{n}} \\
    &\geq \frac{1}{2}\left(\prod_{i =1}^n \big( v_i(g_i^*) \ + \ v_i(\pi(i) + X_i + Y_i) \big) \right)^{\sfrac{1}{n}} \tag{via Lemma \ref{lemma:matchhigh}} \\
    &\geq \frac{1}{2}\left(\prod_{i=1}^n  v_i(g_i^* + \pi(i) + X_i + Y_i) \right)^{\sfrac{1}{n}} \tag{since $v_i$ is {\XOS}, i.e., subadditive} \\
    &= \frac{1}{2} \NSW(\Q^*).
\end{align*}
\end{proof}

\newpage
\bibliographystyle{alpha} 
\bibliography{references.bib}

\end{document}